%% file: main.tex
\documentclass[sigconf]{acmart}
\AtBeginDocument{%
  }

\setcopyright{cc}
\setcctype{by}

\acmJournal{PACMMOD}
\acmYear{2026} \acmVolume{4} \acmNumber{1 (SIGMOD)} \acmArticle{88} \acmMonth{2} \acmPrice{}\acmDOI{10.1145/3786702}

\usepackage{amsmath,amsthm}
\usepackage{graphicx}
\usepackage{booktabs}
\usepackage{soul}
\usepackage[ruled,linesnumbered,vlined]{algorithm2e}
\usepackage{hyperref}
\usepackage{cleveref}
\usepackage{xcolor}
\usepackage{booktabs}
\usepackage{enumitem}
\usepackage{tikz}
\usepackage{subcaption}
\usetikzlibrary{shapes,arrows,positioning,calc}
\usepackage{float}
\usepackage{placeins}
\usepackage{makecell}

\usepackage[most]{tcolorbox}
\tcbset{
  surrogateexample/.style={
    enhanced,
    colback=gray!5!white,   
    colframe=blue!40!black,
    coltitle=black,
    boxrule=0.5pt,
    arc=2mm,
    left=0.7em,
    right=0.7em,
    top=0.18em,
    bottom=0.18em,
    boxsep=0.4mm,
    fontupper=\ttfamily\footnotesize,
    fonttitle=\footnotesize\sffamily\bfseries,
    title filled,
    colbacktitle=blue!6!white,
    attach boxed title to top left={xshift=0.7em,yshift*=-1mm},
    boxed title style={
      sharp corners=south, size=small, top=0.05mm, bottom=0.05mm,
      left=0mm, right=0mm, interior style={fill=blue!6!white}
    },
    drop shadow={black!7!white},
    before skip=4pt, after skip=4pt,
    halign=left,
  }
}

\usepackage{listings}

\input{macros}

\begin{document}

\title{Task Cascades for Efficient Unstructured Data Processing}

\author{Shreya Shankar}
\affiliation{%
  \institution{UC Berkeley}
  \country{USA}
}
\email{shreyashankar@berkeley.edu}

\author{Sepanta Zeighami}
\affiliation{%
  \institution{UC Berkeley}
  \country{USA}
}
\email{zeighami@berkeley.edu}

\author{Aditya G. Parameswaran}
\affiliation{%
  \institution{UC Berkeley}
  \country{USA}
}
\email{adityagp@berkeley.edu}

\begin{abstract}
Modern database systems allow users to query or process unstructured text or document columns using LLM-powered functions. Users can express an operation in natural language (e.g., {\em ``identify if this review mentions billing issues''}), with the system executing the operation on each document, in a row-by-row fashion. One way to reduce cost on a batch of documents is to employ the model cascade framework: a cheap proxy model processes each document, and only uncertain cases are escalated to a more accurate, expensive oracle. However, model cascades miss important optimization opportunities; for example, often only part of a document is needed to answer a query, or other related, but simpler operations (e.g., {\em ``is the review sentiment negative?'', ``does the review mention money?''}) can be handled by cheap models more effectively than the original operation, while still being correlated with it. 

We introduce the {\bf \em task cascades} framework, which generalizes model cascades by varying not just the model, but also the {\em document portion} and {\em operation} at each stage. Our framework uses an LLM agent to generate simplified, decomposed, or otherwise related operations and selects the most relevant document portions, constructing hundreds of candidate tasks from which it assembles a task cascade. We show that optimal cascade selection is intractable via reduction from {\sc Minimum Sum Set Cover}, but our iterative approach constructs effective cascades.  We also provide an extension that offers statistical accuracy guarantees: the resulting cascade meets a user-defined accuracy target (with respect to the oracle) up to a bounded failure probability. 
Across eight real-world document processing tasks at a 90\% target accuracy, task cascades reduce end-to-end cost by an average of 36\% compared to model cascades at production scale.
\end{abstract}


\received{July 2025}
\received[revised]{October 2025}
\received[accepted]{November 2025}

\maketitle

\input{sections/intro}

\input{sections/background}

\input{sections/cascade_construction}
\input{sections/restructuring}
\input{sections/surrogate}

\input{sections/cost_model}
\input{sections/experiments}

\input{sections/related}

\input{sections/conclusion}

\bibliographystyle{ACM-Reference-Format}
\bibliography{sample}

\clearpage
\FloatBarrier
\appendix
\input{sections/app}

\end{document}

%% file: macros.tex
\newcommand{\topic}[1]{\smallskip
\noindent {\bf #1.}}

\newcommand{\ttt}[1]{{\small \texttt{#1}}\xspace}

\newcommand{\method}[1]{{\textbf{\ttt{#1}}}}
\newcommand{\methodsmall}[1]{{\textbf{\scriptsize \texttt{#1}}}}

\definecolor{navyblue}{rgb}{0.0, 0.0, 0.5}

\definecolor{rblue}{RGB}{38, 70, 168}      %
\definecolor{rred}{RGB}{172, 41, 37}       %
\definecolor{rgreen}{RGB}{24, 78, 24}   %
\definecolor{rpurple}{RGB}{126, 47, 142}   %

\newcommand{\all}[1]{#1}
\newcommand{\rone}[1]{#1}
\newcommand{\rtwo}[1]{#1}
\newcommand{\rthree}[1]{#1}

\newcommand{\techreport}[1]{#1}
\newcommand{\cameraready}[1]{}

%% file: sections/intro.tex
\section{Introduction}
\label{sec:intro}

\begin{figure*}
\techreport{\vspace{-10pt}}
\cameraready{\vspace{-5pt}}
\includegraphics[width=0.95\linewidth]{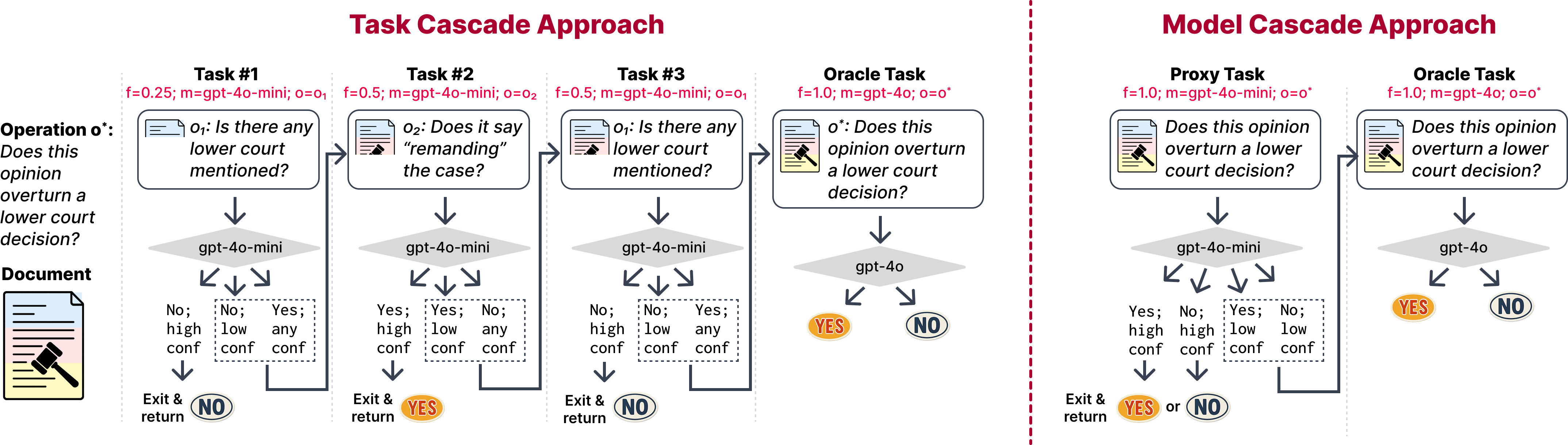}
\vspace{-5pt}
\caption{Task vs. model cascade for determining whether a Supreme Court opinion overturns a lower court decision in \Cref{ex:supreme-court-task}. In a task cascade (left), each stage (or task) is defined by a document fraction ($f$), a model ($m$), and an operation ($o$). Surrogate operations may be reused at different document fractions (as in tasks 1 and 3). The final oracle task applies the original user-specified operation ($o^*$) on the full document ($f=1.0$) with the oracle model ($m=$gpt-4o) if prior tasks cannot confidently resolve the input. In contrast, a model cascade (right) applies the same operation $o^*$ to the full document at each stage, varying only the model.}
\vspace{-8pt}
\label{fig:task-cascade}
\end{figure*}

Large language models (LLMs) are increasingly being integrated as operators in data management systems, enabling users to analyze unstructured text columns more easily. This line of work spans SQL extensions to databases, providing LLM-powered {\em map} and {\em filter} operations~\cite{databricks-llm, duckdb-llm, alloydb-llm,snowflake-llm}, as well as Python-based declarative frameworks for executing pipelines of LLM-powered operators~\cite{liu2024declarative, patel2024lotus, shankar2024docetl, wang2025aop, anderson2024design}. In all cases, users describe an operation in natural language, with the system executing it by invoking an LLM on each text value---hereafter referred to as a {\em document}---in a row-by-row fashion. Consider the following example:

\begin{example}[Supreme Court Opinion Analysis]
\label{ex:supreme-court-task}
Suppose a legal analyst is working with a large collection of Supreme Court opinions. The analyst’s goal is to determine, for each opinion, whether it overturns a lower court decision. In this setting, each {\em document} is a long, complex legal text, and the analyst poses an {\em operation} in natural language: {\em ``Does this opinion overturn a lower court decision?''}
\end{example}

State-of-the-art LLMs like GPT-4o and GPT-4.1 can achieve high accuracy on tasks such as \Cref{ex:supreme-court-task}. However, commercial LLM APIs charge by the number of input and output tokens, making large-scale analysis costly. As a result, users look for ways to cut inference costs without losing much of the accuracy provided by the best models. They are often willing to accept a modest drop in accuracy---e.g., targeting 90\% of GPT-4o's accuracy---in exchange for significant savings~\cite{russo2025abacus, patel2024lotus}. Although cheaper models exist (such as GPT-4o-mini, which is nearly $17\times$ cheaper than GPT-4o), they often deliver much lower accuracy.

\topic{Prior Work and Limitations} The standard approach to reduce costs of map and filter operations with machine learning models, while preserving output accuracy, is to use a {\em model cascade}~\cite{kang2017noscope, anderson2019physical, kang2022tasti, russo2023accelerating, Kang2021AcceleratingAA}; with recent systems all adopting it as-is for LLM-powered operators~\cite{chenfrugalgpt, patel2024lotus, yue2024large}. Here, the system first applies a cheaper {\em proxy} model to each input document. If the proxy's confidence in its prediction exceeds a system-selected confidence threshold, the system accepts the proxy's output. Otherwise, the system escalates the input to a more accurate but more expensive {\em oracle} model (e.g., the best available LLM). Users specify a performance target (e.g., 90\% of the oracle's accuracy), and the system tunes the confidence thresholds to meet this target while minimizing cost. 

However, the model cascades framework, by focusing only on swapping models, underutilizes optimization opportunities. For instance, when no cheaper model can perform the operation accurately, documents will end up being processed by the oracle. In fact, on medical or legal tasks, cheaper models may have as little as 30\% of the accuracy of expensive models~\cite{vals-ai, fan2025lexam,magnini2025open}, with the gap being even wider for long-context reasoning~\cite{wu2025emergence, liu2024lost}. Moreover, model cascades always process the full document, even though many operations require only a small, relevant portion of the text. 

\topic{Change the Task, Not Just the Model} Our key insight is that {\em effective cascades should vary the operation as well as the portion of the document}, in addition to the model. Earlier stages of the cascade can execute any operation that is easier for cheaper models---such as a prerequisite check, a decomposed sub-task, or any related operation that is correlated with the original. We refer to such operations as {\em surrogate} operations: variants of the original operation that are easier for a proxy model to perform.  In \Cref{ex:supreme-court-task}, rather than having a proxy model decide if an opinion overturns a lower court decision, we can have it check whether the opinion mentions a lower court at all. If not, we don't need to invoke the oracle. 

Another complementary strategy to improve proxy model accuracy is to {\em prune irrelevant context} from each document before inference. For example, instead of providing the entire document to the LLM, we can supply only the sections that are semantically relevant to the operation---such as paragraphs containing ``overturn,'' ``judgment vacated,'' or their synonyms. By eliminating unrelated text, the proxy model is more likely to perform the operation both correctly {\em and} with lower cost~\cite{xu2023retrieval, tan2024lloco}. 

In this paper, we introduce the notion of a \textbf{task cascade}: an ordered sequence of {\em tasks}, each task defined by an operation, an LLM, and a selected fraction of the document to analyze. As illustrated in \Cref{fig:task-cascade}, for \Cref{ex:supreme-court-task}, a task cascade may first check if there is any lower court mentioned in a short excerpt before switching over to a different surrogate,  and later revisit this same operation over even more of the document---before finally applying the original operation with the oracle model on the full document if needed. 

\topic{Task Cascade Challenges} However, designing an effective task cascade is challenging for several reasons. First, discovering effective surrogate operations is hard, and there is an infinitely large search space. While one can use an LLM to author a new operation, a surrogate operation is only beneficial if it is both accurately performed by a cheap model {\em and} if it resolves a substantial fraction of documents. A poor surrogate can potentially make the entire cascade {\em less} efficient than if it were omitted. Second, determining the ideal portion of the document for each task is non-trivial. The goal is to minimize token costs by providing just enough context for an accurate decision; too little text risks errors, while too much (like a full document) negates potential savings and can overwhelm cheaper models. Third, determining the best order in which to apply tasks is difficult. The number of possible sequences grows rapidly with more tasks, and each sequence can differ significantly in cost and effectiveness. A task may perform relatively well overall but poorly when applied only to the harder examples that earlier tasks did not resolve, and thus not meet accuracy constraints. Moreover, LLM APIs offer cost savings (50–90\%) when multiple prompts share the same input prefix. As a result, optimizing task order is not just about accuracy and selectivity of tasks: it also involves structuring inputs to exploit reuse. {\bf\em Overall, jointly optimizing surrogate operations, document fractions, and task orders creates a complex search space that is difficult to reliably navigate.}

\topic{Effective Task Cascades} We introduce an approach for constructing effective task cascades, given a document set, user-specified operation, and a target accuracy with respect to an oracle. We address the aforementioned challenges as follows:

\begin{itemize}[nosep, leftmargin=*, wide=0pt]
\item \textbf{Task ordering.} Given a set of tasks, trying all orderings to find the best one quickly becomes infeasible as the number of tasks grows. We show that optimally ordering tasks is {\sc NP-Hard} through a reduction from {\sc Minimum-Sum-Set-Cover}, and inspired by approximation algorithms for this problem, we develop a greedy algorithm that sequentially adds the task that most reduces total inference cost while satisfying accuracy constraints.

\item \textbf{Confidence threshold selection.} Per task in the cascade, we need to set confidence thresholds that determine whether to accept the proxy model output for that task. A naive approach sets thresholds for each task so that, by a union bound, the overall cascade meets the target accuracy with the desired probability. However, as the number of tasks increases, the naive method becomes overly conservative, making it difficult to find a feasible cascade. We develop a new threshold adjustment algorithm, building on recent statistical results~\cite{waudby2024estimating}, which enables us to achieve accuracy guarantees with far fewer samples.

\item \textbf{Surrogate operation generation.} To generate the set of tasks in the first place, a simple idea is to prompt an LLM agent (e.g., based on OpenAI's o1) for candidate operations, but these surrogates are often not effective for smaller models, as the agent has no way to know which operations are actually easy for cheaper models to perform. Instead, we introduce an {\em iterative} approach, where the agent proposes surrogates, and tests and refines them based on which are actually performed well by proxies.

\item \textbf{Document pruning.} Finally, per task, we want to minimize the portion of the document provided as input, while maintaining accuracy. A naive approach to pruning irrelevant portions might use embedding similarity to select document chunks most related to the user operation, similar to standard retrieval-augmented generation (RAG) techniques; however, this approach often fails for complex or long-context tasks, so we instead train a lightweight classifier---supervised by the oracle LLM---to score each chunk's relevance and use these scores for pruning.
\end{itemize}

Overall, our contributions include:
\begin{enumerate}[nosep, leftmargin=*, wide=0pt]
\item We introduce the concept of {\em task cascades} for LLM-based document processing, generalizing model cascades by allowing each stage to specify a model, operation (original or surrogate), and document fraction.

\item We show that the problem of optimal cascade construction is {\sc NP-Hard}, via a reduction from {\sc Minimum Sum Set Cover}.

\item We develop an automated, agentic approach for constructing effective task cascades---jointly discovering surrogate operations, pruning documents, and greedily assembling task sequences that minimize cost while meeting accuracy targets. We show that our approach can be extended to provide guarantees.

\item We evaluate our method on eight complex document classification tasks drawn from Kaggle and prior LLM-based data processing research, comparing against standard model cascade baselines and analyzing the contribution of each component (task ordering, document pruning, surrogate generation).  Across all workloads at a 90\% target accuracy, {\bf \em our approach reduces inference cost by 48.5\% on average compared to model cascade baselines, and by 86.2\% compared to using the oracle model alone}.
\end{enumerate}

We present the problem and an overview of our approach in \Cref{sec:background}. Next, in \Cref{sec:cascade-construction}, we show the hardness of optimal task ordering and describe our greedy cascade assembly algorithm and new procedure for achieving statistical accuracy guarantees. \all{In \Cref{sec:restructuring}, we describe how we prune irrelevant context in documents, followed by our agentic approach for discovering surrogate operations in \Cref{subsec:surrogate-generation}. Then, in \Cref{sec:cost-model}, we describe how to model the cost of task cascades. Finally, we present our evaluation in \Cref{sec:evaluation} and cover related work in \Cref{sec:related}.}

%% file: sections/background.tex
\section{Problem Setup and Approach}
\label{sec:background}

\begin{figure*}
  \centering
  \vspace{-5pt}
  \includegraphics[width=0.85\linewidth]{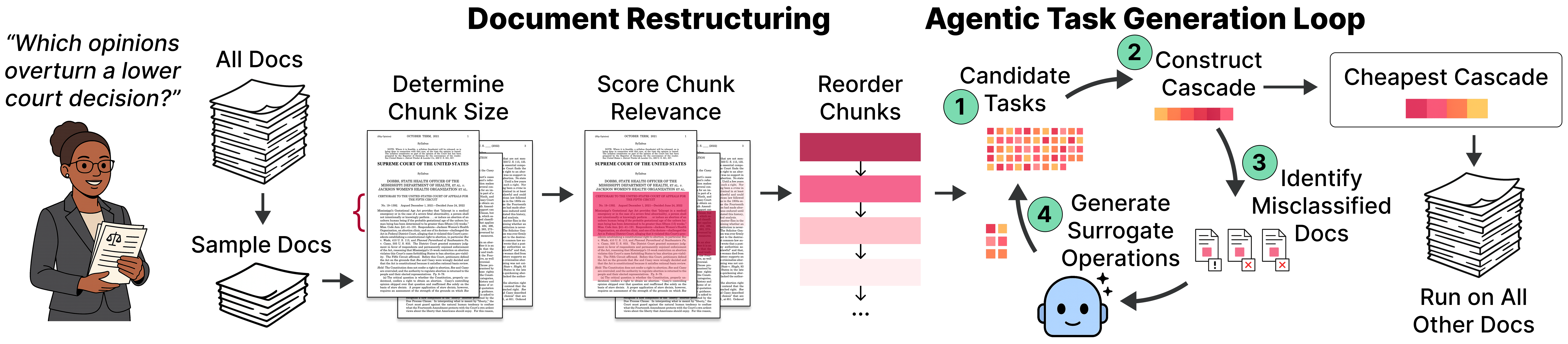}
\vspace{-10pt}
  \caption{Overview of our approach. A user poses a complex classification query over long documents. Our approach restructures each document to prioritize relevant chunks, then iteratively proposes surrogate operations and assembles a cost-effective task cascade to meet accuracy constraints.}
  \label{fig:overview}
  \vspace{-10pt}
\end{figure*}

In this section, we formalize the notion of task cascades and give an overview of our approach to finding an efficient task cascade. A summary of all notation is provided in \Cref{tab:notation}, for convenience.

\subsection{Setup}
\label{sec:background-setup}

\topic{Problem Setting} We focus on the setting where LLM-powered {\em map} and {\em filter} operations in commercial databases~\cite{duckdb-llm, alloydb-llm, databricks-llm, snowflake-llm} or recent systems~\cite{patel2024lotus, liu2024declarative, shankar2024docetl} produce outputs from a fixed set of classes. Formally, the user provides a collection $D$ of documents and a target operation $o_{\mathrm{orig}}$, described in natural language. The goal is to assign each document $x \in D$ to a class $c$ from a predefined set of classes $C$. The user specifies an accuracy target $\alpha \in (0,1]$, requiring the system's predictions to match an oracle model $m_{\mathrm{oracle}}$'s predictions on at least an $\alpha$-fraction of documents.

Given a document $x$ (or, alternatively, a fraction of the document, as described below) and operation $o$ (either the original or a surrogate, defined below), a model $m$ produces a score $p_m(c \mid x, o)$ for each class $c \in C$. The class $c$ with the highest $p_m(c \mid x, o)$ is taken as the prediction of the model. 
For \Cref{ex:supreme-court-task}, $x$ could be a court opinion, $o$ might be {\em ``Does this opinion overturn a lower court decision?''}, and $c \in \{\text{True}, \text{False}\}$. 
To enable early termination in a cascade, we use \textbf{confidence thresholds}: for each class $c \in C$, if the model's confidence $p_m(c \mid x, o) \geq \tau^c$ for some threshold $\tau^c$, the prediction for that document $x$ is accepted; otherwise, $x$ continues to the next stage. 

Instead of always using the full document $x$, we may also use a portion of the document, $x_f$, denoting the top $f$ fraction of $x$, containing the most relevant content for performing $o_{\mathrm{orig}}$. Relevance is determined by a scoring function that ranks document segments by their utility for $o_{\mathrm{orig}}$. So, the confidence threshold now based on $p_m(c \mid x_f, o)$. A \textbf{task} is then defined as $T_i = (m_i, o_i, f_i, \tau_i)$, where:
\begin{itemize}
    \item $m_i$: model (e.g., proxy or oracle)
    \item $o_i$: operation (original or surrogate)
    \item $f_i$: fraction of the document processed
    \item $\tau_i = \{\tau_i^c\}_{c \in C}$: class-specific confidence thresholds
\end{itemize}
A {\bf task cascade} is an ordered sequence of tasks $\pi = (T_1, T_2, \ldots, T_k)$. At inference time, for a given document $x$, we evaluate each $T_i$ in sequence: model $m_i$ is applied to input $x_{f_i}$ with operation $o_i$, yielding a predicted class $c$ and confidence score $p_{m_i}(c \mid x_{f_i}, o_i)$. If this confidence exceeds the corresponding threshold $\tau_i^{c}$, the prediction is accepted, $\mathrm{Cascade}(\pi, x) = c$, and the cascade terminates; otherwise, processing continues to the next task. If none of the $k$ tasks return a confident prediction, the cascade defers to the oracle task $T_{k+1} = (m_{\mathrm{oracle}}, o_{\mathrm{orig}}, 1, \varnothing)$, which applies $o_{\mathrm{orig}}$ to the full document with the oracle, with $\mathrm{Cascade}(\pi, x)$ set to the oracle's prediction. Thus, a document ``leaves'' the cascade at the first confident prediction within $\pi$, or at the oracle if all tasks defer. A traditional model cascade is thus a ``restricted'' task cascade where each $T_i$ uses the same operation $o_i = o_{\mathrm{orig}}$ and processes the entire document ($f_i = 1.0$) at each stage, varying only the model $m_i$.

\topic{Confidence Scores} Each model $m_i$ defines a distribution over output classes $c \in C$, conditioned on operation $o_i$ and document portion $x_{f_i}$, i.e., $p_{m_i}(c \mid x_{f_i}, o_i)$. Inputs to LLMs are tokenized; a \emph{token} is a subword unit such as a word fragment or punctuation mark. LLMs return log-probabilities for each token in the generated output, which can be transformed into confidence scores.

\topic{Cost Model} Each task $T_i = (m_i, o_i, f_i, \tau_i)$ sends a prompt to model $m_i$ by concatenating the document fraction $x_{f_i}$ and the operation $o_i$. We denote the number of tokens or {\em size} as $|\cdot|$, so $|x_{f_i}|$ is the size of the fractional document, and $|o_i|$ is the size of the operation prompt. We let $\lambda_{\mathrm{in}}$ and $\lambda_{\mathrm{cached}}$ denote the cost per new input token and per cached input token, respectively. For classification tasks, the output  cost is typically insignificant and can be ignored, but can be easily incorporated if needed.

The cost to run task $T_i$ on $x_{f_i}$ depends on how much of the input is already cached from previous calls to the same model (as supported by LLM APIs). The cost is:

\begingroup
\setlength\abovedisplayskip{0pt}
\setlength\belowdisplayskip{0.5em}
\begin{align*}
\mathrm{Cost}(T_i, x) = 
\begin{cases} 
    |x_{f_i}|\, \lambda_{\mathrm{in}} + |o_i|\, \lambda_{\mathrm{in}} & \text{if $x_{f_i}$ is new} \\[0.5em]
    |x_{f_j}|\, \lambda_{\mathrm{cached}} + (|x_{f_i}| - |x_{f_j}|)\, \lambda_{\mathrm{in}} + |o_i|\, \lambda_{\mathrm{in}} & \text{if $x_{f_j} \subseteq x_{f_i}$}
\end{cases}
\end{align*}
\endgroup

\noindent where $x_{f_j}$ is the largest previously processed document prefix that is a subset of $x_{f_i}$. 
Throughout, we place the document before the operation in the prompt, maximizing cache utilization of document tokens across multiple tasks using the same model. The total inference cost for document $x$ is then:

\begingroup
\setlength\abovedisplayskip{-5pt}
\setlength\belowdisplayskip{0.5em}
\begin{align*}
\mathrm{Cost}(\pi, x) = \sum_{i=1}^{j^*} \mathrm{Cost}(T_i, x)
\end{align*}
\endgroup

\noindent where $j^*$ is the index of the first task to return a confident prediction (i.e., where $\mathrm{Cascade}(\pi, x)$ is determined), or $j^* = k+1$ if the oracle is called on the original document.

\topic{Problem Statement} Given a document collection $D$, an original operation $o_{\mathrm{orig}}$, a set of classes $C$, an oracle model $m_{\mathrm{oracle}}$, available proxy models $M$, and an accuracy target $\alpha$, we seek to construct a minimal-cost task cascade $\pi = (T_1, T_2, \hdots, T_k)$ as follows:

\begingroup
\setlength\abovedisplayskip{-5pt}
\begin{align*}
    \min_{\pi} \quad & \sum_{x \in D} \mathrm{Cost}(\pi, x) \notag \\
    \text{s.t.} \quad & \Pr\left[\frac{1}{|D|} \sum_{x \in D} \mathbb{I}\big[\mathrm{Cascade}(\pi, x) = T_{k+1}(x)\big] \geq \alpha\right] \geq 1 - \delta
\end{align*}
\endgroup

\noindent where $\mathbb{I}[\cdot]$ is the indicator function and $T_{k+1}$ is the oracle task; i.e., $T_{k+1} = (m_{\mathrm{oracle}}, o_{\mathrm{orig}}, 1, \varnothing)$; so $T_{k+1}(x)$ denotes the oracle's prediction for $x$. $\delta \in (0,1)$ is a user-specified upper bound on the probability that the cascade's accuracy on $D$ falls below $\alpha$.

While we focus on classification with fixed output classes, the task cascade framework can be extended to open-ended map or generation operations, provided there is a reliable automatic evaluator for correctness of an output (i.e., an LLM-as-a-judge~\cite{zheng2023judging} and a means to compute confidence scores, e.g., the geometric mean of per-token probabilities as in~\cite{guptalanguage}.

\subsection{Overview of Our Approach}
\label{sec:background-overview}

\begin{algorithm}[t]
\centering
\begin{minipage}{0.98\linewidth}
\footnotesize
\SetAlgoLined
\KwIn{Documents $D$, original operation $o_{\mathrm{orig}}$, oracle and proxy models}
\KwOut{Task cascade $\pi$ (optionally with statistical guarantees)}
\BlankLine

\tcp{Prepare development set}
Sample $D_{\mathrm{dev}}$ from $D$\;

\tcp{Document Restructuring (\Cref{sec:restructuring})}
\ForEach{document $x \in D_{\mathrm{dev}}$}{
    Train lightweight classifier to score chunk relevance\;
    Reorder $x$ to front-load relevant content\;
}

\tcp{Initialize candidate tasks with original operation at multiple document fractions $\mathcal{F}$}
$\mathcal{T} \gets \{(m, o_{\mathrm{orig}}, f) : m \in \{\text{proxy, oracle}\}, f \in \mathcal{F}\}$\;

\tcp{Agentic Loop for Surrogate Discovery}
\For{$n_a$ iterations}{
    \tcp{Assemble best cascade from current candidates (\Cref{sec:cascade-construction})}
    $\pi \gets$ GreedyCascadeAssembly($\mathcal{T}$, $D_{\mathrm{dev}}$)\;
    
    \tcp{Generate new surrogate operations using agent (\Cref{subsec:surrogate-generation})}
    $\mathcal{O}_{\mathrm{new}} \gets$ LLMAgent($o_{\mathrm{orig}}$, $\pi$)\;
    
    \tcp{Add new candidates to task set}
    $\mathcal{T} \gets \mathcal{T} \cup \{(m, o, f) : o \in \mathcal{O}_{\mathrm{new}}, m \in \{\text{proxy, oracle}\}, f \in \mathcal{F}\}$\;
    
    \If{no cost improvement}{
        \textbf{break}\;
    }
}

\Return $\pi$\;
\caption{\rone{Task Cascades Overview}}
\label{alg:task-cascades-overview}
\end{minipage}%
\end{algorithm}

Our approach constructs a cost-effective task cascade $\pi$ in an offline phase, using a small, representative sample of the document collection $D$, which we'll refer to as $D_{\mathrm{dev}}$, the {\em development set}. \rone{The complete procedure is summarized in \Cref{alg:task-cascades-overview}.}  The construction involves: (i) restructuring documents to support efficient fractional processing, and (ii) an agentic loop that iteratively refines a set of candidate tasks and assembles an optimized cascade. We focus on a two-model setting---a cheap proxy model (e.g., GPT-4o-mini) and a more accurate, expensive oracle (e.g., GPT-4o)---as in model cascades, but our approach can generalize to any number of models. An overview of our approach is depicted in \Cref{fig:overview}.

\subsubsection{Document Restructuring}

We pre-process each document to ensure that the most relevant content for the original operation $o_{\mathrm{orig}}$ appears at the beginning. This reordered layout lets us reuse computation: if a longer document fraction is needed later in the cascade, the tokens from shorter prefixes have already been processed and cached by the model, so we only pay to process the additional new content. We use the oracle model to label regions of each document by relevance to $o_{\mathrm{orig}}$, then train a lightweight classifier to automate reordering for new documents not in $D_{\mathrm{dev}}$. 

\subsubsection{Agentic Loop for Cascade Construction}

To build $\pi$, we iteratively grow a set of candidate task configurations $\mathcal{T}$ and repeatedly assemble the lowest-cost cascade that meets the target accuracy $\alpha$ on $D_{\mathrm{dev}}$. Each candidate task configuration is defined by a model, an operation, and a document fraction---that is, the fraction of the reordered document (starting from the beginning) made available to the model. We use a fixed set $\mathcal{F}$ of fractions (e.g., $\mathcal{F} = \{0.1, 0.25, 0.5, 1.0\}$). Initially, we populate the candidate set $\mathcal{T}$ with all combinations of the original operation $o_{\mathrm{orig}}$, both models, and these document fractions. (Confidence thresholds for each task are determined during cascade assembly.)

Each iteration of the agentic loop proceeds as follows. First, we assemble the lowest-cost cascade, $\pi_{\mathrm{best}}$, that achieves the accuracy target, given the current $\mathcal{T}$. Next, we identify $\pi_{\mathrm{best}}$'s ``failure cases,'' i.e., documents that reach the oracle or are not confidently classified by any task in the cascade. We then prompt an LLM agent to propose new surrogate operations, providing context about these failure cases to inform the agent's proposals. For each new surrogate, we pair it with all models and document fractions and add the resulting tasks to $\mathcal{T}$. We repeat this loop for a fixed number of rounds or until no further cost improvements are found.

\all{Our approach has three main technical components, and we describe them in the remainder of the paper. First, we present cascade assembly (\Cref{sec:cascade-construction}), which operates independently of how tasks in the cascade are produced. We then follow the approach to creating candidate tasks as depicted in \Cref{fig:overview}: document restructuring (\Cref{sec:restructuring}) followed by surrogate generation (\Cref{subsec:surrogate-generation}). These components are highly modular and can each be customized or replaced.}

%% file: sections/cascade_construction.tex
\begin{table}[t]
\centering
\scriptsize
\vspace{-8pt}
\caption{Table of Notation}\label{tab:notation}
\vspace{-6pt}
\techreport{\begin{tabular}{@{}l p{6cm}@{}}}
\cameraready{\begin{tabular}{@{}l p{8cm}@{}}}
\toprule
\textbf{Symbol} & \textbf{Description} \\
\midrule
$x \in D$ & A document in a collection $D$ \\
$D_{\mathrm{dev}}$ & Development set of documents (to find a cascade) \\
$o_{\mathrm{orig}}$ & Original operation, in natural language \\
$c \in C$ & A class from a predefined set of classes $C$ \\
$\alpha$ & Accuracy target, where $\alpha \in (0,1]$ \\
$m;\ m_{\text{oracle}} \in \mathcal{M}$ & Any model; oracle model in the set of models \\
$p_m(c \mid x_f, o)$ & Score produced by $m$ for class $c$, given doc fraction $x_f$ and operation $o$ \\
$\tau^c$ & Class-specific confidence threshold for class $c$ \\
$(m_i, o_i, f_i)$ & A task config., defined by $m_i$, operation $o_i$, doc fraction $f_i$ \\
$T_i = (m_i, o_i, f_i, \tau_i)$ & A task, comprising a task config. and confidence thresholds $\tau_i$ \\
$\pi = (T_1, T_2, \ldots, T_k)$ & A task cascade, i.e., ordered sequence of tasks \\
$Cost(T_i, x)$ & Cost to run task $T_i$ on document $x$ \\
$|\cdot|$ & Number of tokens or size \\
$\lambda_{in},\ \lambda_{cached}$ & Cost per new (cached) input token \\
$f \in \mathcal{F}$ & Set of document fractions to consider for each task \\
$n_s;\ n_a$ & Number of new surrogate operations per iteration; number of iterations \\
$g$ & Minimum fraction of $D_{\mathrm{dev}}$ classified by a task for inclusion in cascade \\
\bottomrule
\end{tabular}
\vspace{-10pt}
\end{table}

\section{Cascade Construction}
\label{sec:cascade-construction}

We describe how to construct a cost-efficient cascade from a set of candidate task configurations $\mathcal{T} = \{(m_i, o_i, f_i)\}$, comprising model, operation, and document fraction tuples. Our objective is to assemble an ordered sequence of tasks, with associated confidence thresholds, that minimizes inference cost while meeting a target accuracy $\alpha$, as formalized in \Cref{sec:background-setup}. For this section, we assume $\mathcal{T}$ is fixed; \all{\Cref{sec:restructuring} will describe how we reorder content within each document, creating fractional documents, and \Cref{subsec:surrogate-generation} will explain how we generate new surrogate operations. Throughout, we assume a fixed set of document fractions per model and operation (original or surrogate).}

We first show that a substantially simplified version of cascade assembly is {\sc NP-hard} (\Cref{ssec:hardness}), motivating our greedy approach. We then detail a three-step procedure for constructing the cascade (\Cref{subsec:assembling-task-cascade}):
(i) filtering out candidate tasks that cannot meet the accuracy requirement;
(ii) greedily assembling an ordered cascade that maintains per-task accuracy; and
(iii) adjusting the cascade to meet the overall accuracy guarantee.

\subsection{NP-Hardness of Optimal Task Cascade}
\label{ssec:hardness}

Constructing
an optimal task cascade $\pi$
on $D_{\mathrm{dev}}$,
given task 
configurations $\mathcal{T} = \{(m_i, o_i, f_i)\}$ is {\sc NP-Hard}.
We show hardness holds even if
all $m_i = m$, a single proxy
model, and all fractions
$f_i = 1$. Our reduction
is from the {\sc NP-Hard} 
{\sc Min-Sum-Set-Cover} problem~\cite{feige2004approximating} ({\sc MSSC}), 
a variant of the more
standard {\sc Set-Cover} problem.
In {\sc MSSC},
we similarly
pick a subset
of sets that cover a universe
of items; however,
the output is an ordering
of these sets,
where, for each item,
we pay a cost based 
on the earliest set 
that covers it; specifically,
if the $i$th set in the sequence
covers item $j$, then
the cost for covering $j$ is $i$.
Our goal is to minimize 
the total cost of covering all
items.

\begin{theorem}
Constructing
an optimal task cascade is {\sc NP-Hard}.
\end{theorem}

\begin{proof} (Sketch)
Given an instance of {\sc MSSC},
$(U, S)$, where
$U$ denotes the items,
and $S$ denotes the sets,
we generate an instance
of cascade assembly
as follows.
For each item $u \in U$, we create a document $d_u$.
For each set $S_i\in\mathcal{S}$, we create a candidate task $T_i = (m, o_i, 1)$, where $m$ is a single proxy model, and $o_i$ is an operation that predicts \textsc{True} with confidence $1$ on $d_u$ iff $u\in S_i$ and returns a random answer with confidence $0$ otherwise.
Therefore, every task in our cascade will have thresholds set to 1.

We design the cost model so that accessing cached document tokens costs 0; each document is cached after the first task that processes it. Thus, the total cost for processing document tokens is constant across all cascades and can be ignored. We set each operation $o_i$ to have cost 1, so running any task $T_i$ on any $d_u$ incurs cost 1. 
We set the accuracy target $\alpha=1$ and assign the oracle infinite cost, so every document must exit the cascade via some task. Under this cost model, processing any $d_u$ through a cascade $(T_{i_1},T_{i_2},\dots)$ incurs a cost equal to the index of the first $T_{i_j}$ with confidence 1 on $d_u$. Hence the cascade cost
equals the MSSC objective.
\end{proof}

\citet{feige2004approximating}
describe a
greedy algorithm
for {\sc MSSC},
which,
at each stage,
picks the set that covers
the most (uncovered) items,
and has an approximation factor
of $4$. Feige et al. 
demonstrate that further
improvements to the approximation
factor are difficult.
We present an analogous greedy
algorithm,
adapted to handle varying
models, operations, costs, document fractions, and accuracies.

\subsection{Assembling a Task Cascade}
\label{subsec:assembling-task-cascade}

At a high level, assembling a task cascade from a fixed candidate set
of task configurations $\mathcal{T}$ consists of the following steps:
\begin{enumerate}[leftmargin=*]
\item \textbf{Threshold selection and filtering.} 
For each candidate task individually, we find the smallest confidence thresholds such that at least $g$\% of $D_{\mathrm{dev}}$ are classified by the task, 
with accuracy at least~$\alpha$. We remove
any task for which both criteria are not met for any class.
\item \textbf{Cascade assembly.} We build a cascade
by greedily adding the task that most reduces total inference cost at each step,
such that each task achieves 
the target accuracy on the subset of documents it classifies.
\item \textbf{Threshold adjustment.} We adjust the thresholds 
in the assembled cascade to ensure that the overall sequence meets the desired 
accuracy.
\end{enumerate}

\noindent We discuss each of these steps in turn.

\subsubsection{Threshold Selection and Filtering}
\label{ssec:threshold-filtering}

Our candidate set $\mathcal{T}$ 
may contain hundreds of task configurations, many of which
have no chance of meeting the target accuracy, e.g., 
tasks based on poor surrogate operations. 
We eliminate these by discarding any task configuration
that individually 
cannot confidently classify enough documents 
at the required accuracy. 
For the remaining, we determine 
confidence thresholds up front; 
these thresholds are used in subsequent cascade assembly.

We apply the procedure in \Cref{alg:find-thresholds}: for each candidate task, we examine its predictions on $D_{\mathrm{dev}}$ and, for each class, find the lowest confidence threshold such that predictions above this threshold achieve the target accuracy~$\alpha$. If no such threshold exists for a class (i.e., the required accuracy cannot be achieved for any threshold), 
we set it to be $\infty$, 
ensuring any documents that are predicted to be this class by the task
do not exit the cascade. If, across all classes, these thresholds allow the task to classify at least a minimum fraction $g$ of $D_{\mathrm{dev}}$, 
we retain the task and record its thresholds; otherwise, we discard it.
In our implementation, we set $g$ to be $10\%$; $g$ could also be set using statistical bounds (e.g., from Hoeffding's inequality), 
but for high accuracy targets, we would require many examples per task (e.g., $>250$ for $\alpha = 0.95$ and $\delta=0.25$). 
We still ensure statistical guarantees on our eventual
cascade, as we will discuss in \Cref{ssec:guarantees}.

\begin{algorithm}
\footnotesize
\SetAlgoLined
\KwIn{Candidate task config. $T = (m, o, f)$, development set $D_{\mathrm{dev}}$, accuracy target $\alpha$}
\KwOut{Per-class thresholds $\{\tau^c\}$, or discard $T$}
\BlankLine
Initialize $\tau^c \gets \infty$ for all $c$; \texttt{total} $\gets 0$\;
\ForEach{class $c$}{
    $P_c \gets$ confidence scores assigned to class $c$ by $T$ on all $x \in D_{\mathrm{dev}}$ where $c$ is the model's predicted class for $x$ (i.e., $\arg\max_{c'} p_m(c' \mid x_f, o) = c$); sorted by confidence\;
    \ForEach{unique $t$ in $P_c$ (asc)}{
        $S \gets \{x \in D_{\mathrm{dev}} : \arg\max_{c'} p_m(c' \mid x_f, o) = c,\ p_m(c \mid x_f, o) \geq t\}$\;
        \If{$|S| > 0$ \textbf{and} accuracy of assigning $S$ to $c \geq \alpha$}{
            $\tau^c \gets t$; \texttt{total} $+= |S|$\;
            \textbf{break}
        }
    }
}
\Return $\{\tau^c\}$ \textbf{if} \texttt{total} $\geq g \cdot |D_{\mathrm{dev}}|$; \textbf{else} discard $T$\;
\caption{Find Thresholds for a Given Task Config.}
\label{alg:find-thresholds}
\end{algorithm}

\subsubsection{Cascade Assembly}
\label{ssec:greedy-assembly}

\rone{After filtering, we build the cascade greedily (detailed pseudocode is in \Cref{alg:greedy-cascade}).}
We start with the empty cascade $\pi_0$ that invokes the oracle on all documents. At each step, we consider inserting each unused candidate task at the end of $\pi$, creating a new candidate cascade $\pi'$. For each such candidate cascade $\pi'$, we execute it on $D_{\mathrm{dev}}$ to compute its  cost, and to measure the accuracy of every task in $\pi'$ on the subset of documents it classifies (i.e., those reaching that task, with the output having confidence above its threshold). We only consider $\pi'$ where every task, including the newly added one, achieves the target accuracy $\alpha$ on its assigned subset. Among those that meet the per-task accuracy requirement, we select the one with the lowest total cost. If no such candidate cascade reduces the cost relative to the current cascade, we return the current cascade.

By requiring every task in the cascade to independently satisfy the target accuracy on its assigned documents, we may construct a more conservative cascade than strictly necessary. However, since $D_{\mathrm{dev}}$ may be small, a cascade that merely fits the overall accuracy constraint on the development set may not generalize at inference time. Enforcing per-task accuracy thus increases the likelihood that the cascade meets the desired target. \input{revisions/r1m3}

\subsubsection{Meeting Accuracy Targets with Guarantees}
\label{ssec:guarantees}

\begin{algorithm}[h]
\footnotesize
\SetAlgoLined
\KwIn{Cascade thresholds $\mathcal{T}_0$ and validation set $D_{\mathrm{V}}$}
\KwOut{Cascade thresholds guaranteed to meet the target accuracy}
\BlankLine

\For{iteration $i$ until maximum iteration $max\_iter$}{
    $\mathcal{T}_i\leftarrow \{\tau_j+\epsilon_j^i;\forall\tau_j\in\mathcal{T}_0\}$
    
    \If{not $\mathcal{E}(\mathcal{T}_i, D_{\mathrm{V}})$}{
        \If{i > 1}{
            \Return{$\mathcal{T}_{i-1}$}\;            
        }
        \Else{
            \Return{Not Found}\;                    
        }
    }
}
\Return{$\mathcal{T}_{max\_iter}$}\;            

\caption{Threshold Adjustment}
\label{alg:threshold_adjustment}
\end{algorithm}

The cascade $\pi$ assembled so far is only empirically accurate on the development set and may not achieve the desired accuracy $\alpha$ on new data. To obtain statistical guarantees, we randomly partition the development set $D_{\mathrm{dev}}$ into two i.i.d. splits of equal size: a ``training'' split $D_T$, used to construct the cascade $\pi$, and a ``validation'' split $D_V$, used to adjust thresholds and yield a final cascade $\pi^*$ such that $\Pr[\mathrm{Acc}(\pi^*) < \alpha] \leq \delta$. 
\input{revisions/r1m2}
The new task cascade $\pi^*=(T_1^*, ..., T_k^*)$ is composed of tasks $T_i^*=(m_i, o_i, f_i, \tau_i^*)$ where for all tasks $T_i=(m_i, o_i, f_i, \tau_i)$, $m_i$, $o_i$ and $f_i$ are kept as is, but the thresholds $\tau_i$ are adjusted to $\tau_i^*$. For convenience, we will refer to $\mathbf{t} = \{\tau;\tau\in \tau_i^C,\;\forall C\;\forall i\}$ as the set of all thresholds in all tasks in $\pi$. 

At a high level, determining thresholds that guarantee the target accuracy proceeds as follows.  We start by incrementing each threshold in $\mathbf{t}$ by a large non-negative offset $\epsilon_j^{(1)}$, producing an initial, highly conservative set of thresholds. We then iterate through (1) {\em adjusting} this offset, and (2) {\em estimating} whether the resulting thresholds are sufficient. In the adjustment phase, at each iteration $i$, we reduce the offsets to create a new candidate threshold set $\mathbf{t}^{(i)} = {\tau_j + \epsilon_j^{(i)} : \tau_j \in \mathbf{t}}$, where the adjustment strategy $\epsilon_j^{(i)}$ is chosen to decrease with each step; i.e., $\epsilon_j^{(i + k)} \leq \epsilon_j^{(i)}$ for $k \geq 0$. We will discuss our exact adjustment strategy in \Cref{ssec:threshold-adjustment-strategy}.  In the estimation phase, we run the cascade with $\mathbf{t}^{(i)}$ on the validation split $D_V$ and apply an estimation function $\mathcal{E}(\mathbf{t}^{(i)}, D_V)$. This function returns \ttt{True} if the cascade achieves accuracy at least $\alpha$ on $D_V$ with failure probability at most $\delta$, and \ttt{False} otherwise. This procedure is presented in \Cref{alg:threshold_adjustment}, where the two steps of {\em adjust} and {\em estimate} are repeated until a possible maximum number of iterations, and we stop at the first instance where the estimation function determines that the thresholds do not meet the target. If no candidate set passes, we revert to the oracle-only cascade (which never happens in our experiments).

\begin{theorem}\label{thm:guarantees}
For any adjustment strategy $\epsilon$ and appropriate estimator $\mathcal{E}$, \Cref{alg:threshold_adjustment} returns a cascade $\pi^*$ with $\Pr(\mathrm{Acc}(\pi^*) < \alpha) \leq \delta$.
\end{theorem}

\Cref{thm:guarantees} shows that \Cref{alg:threshold_adjustment} yields the desired accuracy guarantee for any valid estimator and monotonic adjustment schedule. A valid estimator $\mathcal{E}$ must provide a statistical test that, given a candidate set of thresholds and a validation set, returns \ttt{True} only if it can certify---based on the observed sample---that the cascade achieves accuracy at least $\alpha$ with failure probability at most $\delta$. This is a classic problem of mean estimation from i.i.d. Bernoulli samples, and any sound concentration bound can be used; e.g., Hoeffding's inequality. We adopt the estimator from~\citet{waudby2024estimating}, which uses both the sample mean and variance to produce tighter bounds---particularly when the variance is small---compared to Hoeffding's inequality, which relies only on the mean. The proof of \Cref{thm:guarantees} and the full specification of $\mathcal{E}$ are given in \Cref{app:guarantees}\cameraready{ in our technical report~\cite{shankar2026task}}. We also account for multiple applications of $\mathcal{E}$ within the adjustment loop, ensuring the total probability of failure remains bounded by $\delta$.

\subsubsection{Threshold Adjustment Strategy} 
\label{ssec:threshold-adjustment-strategy}

We now describe our strategy for adjusting thresholds to ensure the final cascade meets the desired accuracy guarantee. A naive approach is to decrement each threshold by a fixed amount at every iteration (e.g., $-0.01$ per step). However, LLM-generated confidence scores are often poorly calibrated and heavily concentrated near 1, making uniform decrements ineffective~\cite{jiang2021can, kadavath2022language, desai2020calibration}. Instead, for each class $c$ of each task $T = (m, o, f, \tau)$, we collect all confidence scores greater than the initial threshold $\tau^c$ assigned by the model to predictions of class $c$ on the training split $D_T$. We sort these confidence scores in ascending order: $\tau^c < p_1 < p_2 < \cdots < p_k$, where $p_1$ is the lowest confidence just above $\tau^c$ and $p_k$ is the highest.

We introduce a \emph{shift index} $s$ that starts at $s_{\max}$ and decrements by 1 at each iteration. At each iteration, for every class, we set its threshold to $p_s$, corresponding to a more conservative cutoff when $s$ is large, and a less conservative one as $s$ decreases. Specifically, we start with $p_{s_{\max}}$ (the most conservative threshold), and at each subsequent iteration decrement $s$ by 1, thereby shifting the threshold closer to the original $\tau^c$. Once $s = 0$, we use the original threshold $\tau^c$. After updating all thresholds in this way for an iteration, we evaluate the cascade on $D_V$ and apply the estimator $\mathcal{E}$ to check the accuracy guarantee. If the guarantee is met, we continue; if not, we return the thresholds from the previous shift. Our complete algorithm is given in \Cref{app:guarantees-adjustment}\cameraready{ in our technical report~\cite{shankar2026task}}. In our implementation, we set $s_{\max} = 5$, but values between 3 and 10 are generally effective.

%% file: revisions/r1m3.tex
\rone{Note that for simplicity, this greedy assembly algorithm operates on tasks with fixed, pre-determined confidence thresholds (set in \Cref{ssec:threshold-filtering}). The subsequent threshold adjustment step (\Cref{ssec:guarantees}) is purely a post-processing step to provide statistical accuracy guarantees, and does not further optimize cascade cost.}

%% file: revisions/r1m2.tex
\rone{If $D_T$ or $D_V$ are too small or highly variable, the procedure may simply fail to 
find any cascade meeting the target accuracy under the specified bound; 
accordingly, our experiments use at least 75 documents in each split.}

%% file: sections/restructuring.tex
\all{\section{Document Restructuring}
\label{sec:restructuring}}

\begin{figure}[h]
\centering
\vspace{-5pt}
\resizebox{\ifdim\columnwidth=\textwidth 0.65\textwidth\else 0.85\columnwidth\fi}{!}{
\begin{tikzpicture}[
  node distance=0.6cm and 0.4cm,  %
  box/.style={draw, rounded corners=2pt, align=left, font=\footnotesize, inner sep=6pt, text width=3.8cm, fill=gray!5},
  arrowlabel/.style={font=\scriptsize, fill=white, inner sep=1pt},
  highlight/.style={fill=blue!10}
]
\node[box, text width=3.5cm] (original) {  %
\textbf{Original:}\\
The patient was admitted to the emergency department complaining of severe chest pain and shortness of breath...
};
\node[box, right=of original, text width=5.5cm, font=\scriptsize] (standardized) {  %
\textbf{Standardized:}\\
\texttt{Line \#1. The patient was}\\
\texttt{Line \#2. admitted to the emergency}\\
\texttt{Line \#3. department complaining of}\\[0.3em]
\begin{tcolorbox}[
    colback=blue!10, boxrule=0pt, left=0pt, right=0pt, top=0pt, bottom=0pt,
    enhanced, sharp corners, width=\linewidth
    ]
    \texttt{Line \#4. severe chest pain and}\\
    \texttt{Line \#5. shortness of breath...}
\end{tcolorbox}
};

\draw[->, thick] (original.east) -- (standardized.west)
  node[midway, above, font=\scriptsize, fill=white, inner sep=1pt] {};

\node[box, anchor=north, text width=9.2cm, fill=blue!5] (oracle) 
  at ($(original.south west)!0.5!(standardized.south east) + (0,-0.4cm)$) {
\textbf{Oracle Prompt:} "Identify line ranges needed for the operation: \textit{Does this note indicate an adverse physical reaction?}"\\[0.3em]
\textbf{Oracle Response:} {\scriptsize \texttt{{"ranges": [\{"start\_line": 4, "end\_line": 5\}]}}}
};

\end{tikzpicture}
} %
\caption{Standardizing documents. This format allows the oracle to specify relevant content as line ranges.}
\label{fig:standardization}
\techreport{\vspace{-10pt}}
\end{figure}
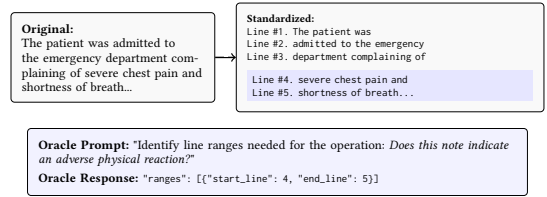

Tasks in a cascade can process different fractions of each document, allowing models to focus on the most relevant content for the user-defined operation. To reduce cost under our prefix-based cost model (\Cref{sec:background-setup}), we reorder each document so that the most relevant text appears first. This restructuring step is independent of the choice of surrogate operations or cascade configuration; any method for ranking document segments by relevance can be used.
\input{revisions/r1m5}

To restructure documents effectively, we first choose how to segment them. Some user operations require only short spans of text (e.g., a sentence), while others need broader context. We represent each document as a sequence of lines and select a segmentation {\em granularity}---the number of consecutive lines to treat as a unit. Each document is then divided into contiguous {\em chunks} of this granularity. Then, we train a lightweight classifier to score and rank chunks in a document by their relevance to the user-defined operation. The following sections describe how we determine the chunk granularity and train the relevance classifier.

\topic{Determining Chunk Granularity} To select the chunk granularity for the user-specified operation $o_{\mathrm{orig}}$, we use the oracle LLM and the development set $D_{\mathrm{dev}}$:

\begin{enumerate}[nosep, left=0pt]
\item We split each document into lines of 80 characters, each prefixed with a line number (see \Cref{fig:standardization}).
\item For each document, we ask the oracle LLM to return the minimal set of non-overlapping line ranges needed to answer $o_{\mathrm{orig}}$. The oracle returns these line ranges in a structured format (e.g., \ttt{{"start\_line": 5, "end\_line": 12}}).
\item We merge overlapping or adjacent line ranges and create a reduced document containing only the union of these lines. We run the oracle LLM on this reduced version and check whether its answer matches the answer on the full document.
\item If, across $D_{\mathrm{dev}}$, this reduced document yields the correct answer for at least an $\alpha$ fraction of documents, we move to (5). Otherwise, we expand each range by one line before the start and one line after the end (while staying within document boundaries), merge again, and repeat the check. This process continues until the target accuracy is achieved or a maximum of $e = 3$ expansions.
\item We set the chunk granularity to the average length, in lines, of the final merged ranges across all documents in $D_{\mathrm{dev}}$.
\end{enumerate}

\noindent Consider the following example for a document $x$ to illustrate the expansion and merging procedure. Suppose the oracle initially selects the ranges \ttt{[23,25]} and \ttt{[28,30]}, with average length $(3 + 3) / 2 = 3$. If the answers produced by running $o_{\mathrm{orig}}$ on the selected lines do not match the answers on the full document for at least an $\alpha$ fraction of documents in $D_{\mathrm{dev}}$, we expand each range for $x$ by one line at both ends, resulting in \ttt{[22,26]} and \ttt{[27,31]}, with average length $(5 + 5) / 2 = 5$. If another expansion is needed, we obtain \ttt{[21,27]} and \ttt{[26,32]}, which now overlap and are merged into a single range, \ttt{[21,32]}, with average length $12$.

\topic{Scoring Chunks by Relevance} Once we determine the segmentation granularity $s$, we divide each document into contiguous sets of $s$ lines---referred to as {\em chunks}. Our goal is to train a lightweight {\em relevance classifier} that predicts whether a chunk contains information required to perform $o_{\mathrm{orig}}$.

\noindent We first partition $D_{\mathrm{dev}}$ into ``training'' ($D_{\text{train}}$) and ``held-out test'' ($D_{\text{test}}$) splits for constructing and evaluating the classifier. From the previous step, we already have starting lines of relevant chunks for all documents. We construct an oracle-labeled dataset for training the relevance classifier as follows:

\begin{enumerate}[nosep, left=0pt]
\item For each starting line identified by the oracle, we extract the $s$-line chunk beginning at that line and label it as relevant ($r = 1$). To construct irrelevant examples, we slide a non-overlapping window of $s$ lines from the start of the document, and include a chunk as irrelevant ($r = 0$) only if it does not overlap with any relevant chunk. Chunks that overlap with any relevant chunk are excluded from the irrelevant set.
\item For each document in $D_{\text{train}}$ and $D_{\text{test}}$, we collect all labeled relevant and irrelevant chunks as described above. The training and test datasets are formed by taking the union of all labeled chunks from each document in $D_{\text{train}}$ and $D_{\text{test}}$, respectively.
\end{enumerate}

Each chunk $c$ in $D_{\text{train}}$ and $D_{\text{test}}$ is embedded using OpenAI's \ttt{text-embedding-3-small}, producing pairs $(\mathrm{embedding}(c), r)$. Because relevant chunks are rare, we upsample the relevant class during training. We fit a logistic regression model with weights initialized to the embedding of $o_{\mathrm{orig}}$ to predict the relevance for each chunk. The model is trained to minimize binary cross-entropy loss on the chunks from $D_{\text{train}}$, using the Adam optimizer and early stopping based on F1 score on the chunks from $D_{\text{test}}$.

At inference time, a new document is partitioned into consecutive, non-overlapping chunks of $s$ lines; each chunk is embedded and assigned a probability of relevance by the trained model. Chunks are then sorted by predicted probability (from highest to lowest) and concatenated to yield the reordered document.

%% file: revisions/r1m5.tex
\rone{While document restructuring is optional, skipping it forgoes the prefix-caching 
discounts that substantially reduce cost.}

%% file: sections/surrogate.tex
\all{\section{Surrogate Operation Generation}
\label{subsec:surrogate-generation}}

So far, we have described how to assemble a cost-optimal cascade from a fixed set of candidate task configurations. We now address the question of {\em how to create this candidate set}. \all{We assume that all documents have already been reordered so that content most relevant to the original operation appears at the beginning (see \Cref{sec:restructuring} for details), and we only discuss how to generate surrogate operations.} At a high level, our approach proceeds as follows. We initialize $\mathcal{T}$ to include all combinations of $o_{\mathrm{orig}}$, each model $m \in \mathcal{M}$, and each fraction $f \in \mathcal{F}$. We then enter a loop that, at each round, (1) assembles the best cascade from the current $\mathcal{T}$, (2) analyzes the failure cases of this cascade, and (3) expands $\mathcal{T}$ using an LLM agent. We next explain what constitutes a valid surrogate operation and how surrogate operations are elicited from the agent, and finally, how the agentic loop is implemented.

\topic{Valid Surrogate Operations} In order to generate surrogate operations, we require some notion of validity. A valid surrogate operation is any operation whose outputs form a subset of the original operation’s classes. For classification tasks, a valid surrogate operation is simply any operation whose output space is a subset of the original classes (optionally including a special ``none of the above'' or $-1$ label). Suppose the original task is to classify biomedical abstracts into four classes: 0 (Research Article), 1 (Review Article), 2 (Clinical Trial), and 3 (Case Report). A surrogate operation could be: {\em ``If the abstract contains phrases like ``randomized trial'' or ``double-blind,'' output 2 (Clinical Trial); otherwise, output $-1$.''} \space Such heuristics enable proxy models to confidently resolve certain examples with simple patterns, even if they cannot perform full multi-class classification.

\topic{Eliciting Surrogate Operations from the Agent} Using the aforementioned notion of surrogate operations, whenever we query the LLM agent for surrogates, we include the following context:
\begin{enumerate}[leftmargin=*]
    \item The user-defined operation, $o_{\mathrm{orig}}$.
    \item ``Failure cases,'' or examples of documents that are {\em not} classified by any existing proxy task in the current best cascade (i.e., documents that reach the oracle). To ensure we can fit multiple examples of documents within the agent's prompt window, we use the oracle model to identify and concatenate only the most relevant snippets from each document---those that support or explain the oracle’s prediction---rather than the full document.
    \item For each task in the current cascade: summary statistics including whether it was selected into the cascade, its document coverage (i.e., number of documents it resolves), and up to 10 ``hard examples,'' or borderline cases where the task predicted incorrectly with high confidence, typically those with confidence just above the acceptance threshold.
    \item Explicit instructions to propose new surrogate operations that (i) are simpler than the original, (ii) output any subset of the original classes, or (iii) target weaknesses in the current cascade. The prompt includes concrete examples and specifies an output format (\ttt{SURROGATE PROMPT: ...}, \ttt{RATIONALE: ...}), so we can parse the surrogate operations.  See \Cref{app:surrogate-operation-gen}\cameraready{ in our technical report~\cite{shankar2026task}} for details.
\end{enumerate}

\topic{Agentic Loop} Now that we have a mechanism for eliciting surrogate operations from the agent, we describe the full agentic loop. At each round:
(1) We assemble the best cascade using the current candidate set, via the procedure in \Cref{subsec:assembling-task-cascade};  
(2) We update the agent's prompt with the new round's failure cases and summary statistics; and  
(3) The agent proposes $n_s$ new surrogate operations. Each surrogate is paired with every model and document fraction, resulting in $|\mathcal{M}| \times |\mathcal{F}| \times n_s$ new candidate tasks, which are added to $\mathcal{T}$.  
We repeat this process until $n_a$ rounds have completed or no further cost improvement is observed.  
Parameter choices (e.g., $n_s$, $n_a$, document fractions, agent model) are flexible (see \Cref{sec:background}); we $n_s=3$, $n_a=3$, and OpenAI's o1-mini as the agent. See \Cref{app:surrogate-operation-gen}\cameraready{ in our technical report~\cite{shankar2026task}} for the full algorithm.

%% file: sections/cost_model.tex
\all{\section{Cost Model for Optimization}
\label{sec:cost-model}}

\all{In this section, we derive the optimization cost for constructing task cascades and show how hyperparameter choices can control the optimization budget. The total optimization cost comprises three components: document restructuring, surrogate operation generation and evaluation, and agent loop costs. For simplicity, we omit candidate task output token costs (which are constant and negligible---e.g., 1 token for binary classification, at most a few tokens for multi-class tasks). Note that in practice, the actual optimization cost may be lower than our derivation suggests: LLM API providers employ prompt caching, and evaluating multiple candidate tasks on the same documents can lead to cache hits that reduce costs. Our optimization implementation does not explicitly optimize for cache reuse, so the formulas below represent a conservative estimate on the true optimization cost.}

\topic{Initial labeling and document restructuring} 
\all{We first run the oracle model once on the full development set to obtain ``ground truth'' labels for computing surrogate task accuracy: $C_{\text{labels}} = N(L + P)\lambda_o$, where $N$ is the number of development documents, $L$ is the average number of tokens per document, $P$ is the number of prompt tokens per call, and $\lambda_o$ is the oracle model price per token.}

\all{Next, we perform document restructuring to prepare documents for fractional processing. We run the oracle model once more to identify relevant line ranges or chunks of the document. We also generate embeddings for all chunks:
\[
C_{\text{doc}} = C_{\text{labels}} + N(L+P)\lambda_o + NL\lambda_{\text{emb}} = N(L+P)(2\lambda_o) + NL\lambda_{\text{emb}}
\]
where $\lambda_{\text{emb}}$ is the embedding price per token.}

\topic{Surrogate operation generation and evaluation.} 
\all{Over $n_a$ iterations, the agent proposes $n_s$ new surrogate operations per iteration, yielding $n_s n_a$ total surrogates. Each surrogate is paired with each document fraction $f \in \mathcal{F}$ and evaluated on the development set using both the proxy and oracle models. Since we evaluate at $|\mathcal{F}|$ different fractions, the total fractional content processed per document is $S_f = \sum_{i=1}^{|\mathcal{F}|} f_i$ (e.g., for $\mathcal{F} = \{0.15, 0.25, 0.5, 1.0\}$, we have $S_f = 1.90$). With $\lambda_p$ the proxy model price per token, the cost for all surrogate evaluations is:
\[
C_{\text{eval}} = N n_s n_a\left[L S_f (\lambda_o+\lambda_p) + P |\mathcal{F}| (\lambda_o+\lambda_p)\right]
\]
\vspace{-15pt}}

\topic{Agent loop} 
\all{The LLM agent (o1-mini) generates new surrogates over $n_a$ iterations, consuming $T_{\text{agent}}$ input tokens and producing $O_{\text{agent}}$ output tokens per round, with prices $\lambda_{o1,\text{in}}$ and $\lambda_{o1,\text{out}}$. Notably, this cost is independent of the development set size $N$ and thus constant: $C_{\text{agent}} = n_a(T_{\text{agent}}\lambda_{o1,\text{in}} + O_{\text{agent}}\lambda_{o1,\text{out}})$.}

\topic{Total optimization cost} \all{The total optimization cost is $C_{\text{opt}} = C_{\text{doc}} + C_{\text{eval}} + C_{\text{agent}}$. The dominant component is $C_{\text{eval}}$, the cost of evaluating candidate surrogate tasks on the development set. This cost can be reduced by: (1) decreasing $n_s$ or $n_a$ to generate fewer surrogates, (2) using fewer document fractions (smaller $|\mathcal{F}|$), (3) excluding the oracle model when generating candidate tasks (replacing $\lambda_o + \lambda_p$ with $\lambda_p$ in $C_{\text{eval}}$), or (4) reducing the development set size $N$, though (4) makes statistical guarantees harder to achieve.  In our experiments, we also evaluate a cheaper variant using strategy (3).}

%% file: sections/experiments.tex
\section{Evaluation}
\label{sec:evaluation}

We design our experiments to answer the following questions:\footnote{Code available at \href{https://github.com/ucbepic/task-cascades}{\textcolor{blue}{https://github.com/ucbepic/task-cascades}}.}
\begin{enumerate}[nosep,leftmargin=1em]
    \item[\textbf{Q1}] Do task cascades reduce inference cost compared to model cascades, and which components of our approach contribute most to these savings?
    \item[\textbf{Q2}] How do cost and accuracy change as target accuracies vary?
    \item[\textbf{Q3}] How consistent are task cascade costs and accuracy across runs?
    \item[\textbf{Q4}] \all{What is the cost of building task cascades?}
\end{enumerate}

Across all workloads at a 90\% target accuracy, {\bf \em our approach reduces inference cost by 48.5\% on average compared to model cascade baselines (with comparable guarantees), and by 86.2\% compared to using the oracle model alone}, highlighting the effectiveness of task cascades (\textbf{Q1}).  Every component of our approach (i.e., surrogate operation discovery, document pruning, and cascade ordering) contributes to cost reduction. For \textbf{Q2}, task cascades consistently improve the cost–accuracy tradeoff as the target accuracy varies. The largest gains occur at lower target accuracies, likely because it is easier to identify surrogate operations that satisfy relaxed accuracy targets. On more challenging workloads, the performance gap between model cascades and task cascades narrows at higher accuracy targets.  \rone{For \textbf{Q3}, we find that task cascades and model cascades both exhibit high variance, but task cascades consistently achieve lower mean and median costs.} \all{Finally, for \textbf{Q4}, we show that optimization costs are negligible at scale, and that even considering optimization cost, task cascades are cheaper than model cascades.}

We first describe our experimental setup (\Cref{sec:setup}). Next, we present results for each research question: cost-effectiveness across variations (\Cref{subsubsec:cost-effectiveness}), performance under varying accuracy targets (\Cref{subsubsec:acrosstargets}), consistency across repeated runs (\Cref{subsubsec:stability}), \all{and optimization costs and break-even analysis (\Cref{subsubsec:optimization-costs}). Finally, we provide insights from analyzing task cascades, including practical deployment considerations and parameter selection guidance (\Cref{sec:qualitative-insights}).
We also detail examples of surrogate tasks and the number of tasks in each cascade in \Cref{app:examples-surr}\cameraready{ in our technical report~\cite{shankar2026task}}.}

\subsection{Setup}
\label{sec:setup}

\subsubsection{Datasets and Workloads} We evaluate our approach on eight document classification workloads chosen to span a range of domains, document lengths, and task complexities (\Cref{tab:dataset-stats}). Five are drawn from prior research on LLM-powered data processing systems. To broaden coverage, we introduce three additional tasks from Kaggle that require multi-step or domain-specific reasoning.

\begin{table}
\centering
\techreport{\vspace{-5pt}}
\footnotesize
\caption{
\all{Workloads with average word counts and corpus sizes. ``\# Docs'' reports the approximate number of documents in the dataset at release time; actual corpus sizes may have changed in later versions.}
}
\label{tab:dataset-stats}
\vspace{-10pt}
\techreport{\begin{tabular}{@{}l p{4.5cm} r r@{}}}
\cameraready{\begin{tabular}{@{}l p{9cm} r r@{}}}
\toprule
\textbf{Dataset} & \textbf{Description of $o_{\mathrm{orig}}$} & \textbf{Avg.\ words} & \all{\textbf{\# Docs}} \\
\midrule
\textsc{agnews}       & Classify news article summaries into one of four topics: \emph{World}, \emph{Sports}, \emph{Business}, or \emph{Science/Tech}. & \textasciitilde37 & \textasciitilde128k \\
\textsc{court}        & Determine if a U.S.\ Supreme Court opinion reverses the lower-court ruling. & \textasciitilde3.7k & \textasciitilde36k \\
\textsc{enron}        & Identify emails sent by C-suite or VP-level executives in the Enron corpus. & \textasciitilde1.5k & \textasciitilde500k \\
\textsc{fever}        & Decide whether a natural-language claim is supported by the provided evidence snippets. & \textasciitilde5.1k & ~185k \\
\textsc{games}        & Determine whether a review praises a different game more than the one being reviewed. & \textasciitilde1.1k & \textasciitilde6.4M \\
\textsc{legal}        & Detect covenants not to sue or IP no-challenge clauses in license agreements. & \textasciitilde8.0k & 510 \\
\textsc{pubmed}       & Classify biomedical articles into one of six study types: \emph{RCT}, \emph{Observational}, \emph{Meta-analysis}, \emph{Bench/Lab}, \emph{Computational}, or \emph{Review}. & \textasciitilde3.1k & \textasciitilde133k \\
\textsc{wiki\_talk}   & Predict whether a Wikipedia Talk-page discussion culminates in an edit revert. & \textasciitilde0.9k & \textasciitilde125k \\
\bottomrule
\end{tabular}
\vspace{-15pt}
\end{table}

\topic{From Prior Work} Workloads (i)-(ii) are derived from
prior work on model cascades~\cite{chenfrugalgpt,patel2024lotus}, while (iii)-(v) are from work on LLM-powered data processing~\cite{liu2024declarative, shankar2024docetl}: (i) \textsc{agnews}~\cite{del2005ranking}, from~\cite{chenfrugalgpt}, is a four-way classification task with extremely short documents. We include \textsc{agnews} as an example of a simple workload where the proxy already meets the accuracy target, to test whether task cascades can provide additional benefits. (ii) \textsc{fever}~\cite{thorne2018fever}, from~\cite{patel2024lotus}, is a claim verification task: given a natural-language claim and a set of evidence snippets, the LLM must decide whether the claim is supported. Since each document involves different claims and evidence, reusable filters and surrogate tasks are difficult to learn, making \textsc{fever} particularly challenging for our approach.
(iii) \textsc{enron}~\cite{klimt2004enron}, from~\cite{liu2024declarative}, consists of emails from the Enron corpus. \input{revisions/r3m4}
(iv) \textsc{legal}~\cite{hendrycks2021cuad}, from~\cite{shankar2024docetl}, tests detection of specific legal clauses in contracts; we convert the original span extraction task into binary classification by focusing on the presence or absence of a single clause type. (v) \textsc{games}~\cite{sobkowicz2016steam}, from~\cite{shankar2024docetl}, is adapted to a more challenging classification task: determining whether a video game review praises another game more than the one being reviewed.

\topic{New Workloads} We introduce three additional workloads from Kaggle and recent NLP datasets. \textsc{court}~\cite{gqfiddler_scotus} involves classifying U.S. Supreme Court opinions and requires multi-step reasoning over lengthy, complex texts. \textsc{wiki\_talk}~\cite{danescu2012echoes,chang2020convokit} challenges models to predict whether a Wikipedia Talk-page discussion culminates in an edit revert, capturing dynamics of online discourse. Finally, \textsc{pubmed}~\cite{cohan-etal-2018-discourse} is a multi-class classification task on long biomedical articles, testing our approach under both large document size and an expanded label space.

For each workload, we sample 1,000 documents (200 for development, 800 for test), except \textsc{legal}, which contains only 509 documents (150 development, 359 test). \input{revisions/r3m2}
All workloads represent binary classification tasks except \textsc{agnews} and \textsc{pubmed}, which are multi-class. Prompts for all workloads are in \Cref{app:task-prompts}\cameraready{ in our technical report~\cite{shankar2026task}}.

\subsubsection{Baselines}
\rtwo{For all experiments in \textbf{Q1--Q4}, we compare task cascades against three reference approaches.}
\method{Oracle Only} runs the oracle model (GPT-4o) on every document, giving an ideal but costly upper bound. 
Our second baseline, \method{2-Model Cascade}, uses the model cascade approach from~\cite{patel2024lotus}, pairing GPT-4o-mini as proxy with GPT-4o as oracle. 
This approach is state-of-the art as it leverages token-level
log probabilities from LLM APIs, unlike prior model
cascade papers that employ more heavy-weight and less
accurate approaches~\cite{chenfrugalgpt, yue2024large}. 
For each class, we set the proxy threshold to the smallest value on the development set such that the combined accuracy of proxy predictions above the threshold {\em and} oracle predictions below it meets the target accuracy, aggressively minimizing cost. 
Third, \method{2-Model Cascade (+ Guarantees)} augments \method{2-Model Cascade} with our statistical guarantee procedure (\Cref{alg:threshold_adjustment}). 
Although alternatives such as SUPG~\cite{kang2020approximate} could enforce guarantees, we apply the same procedure we do to ensure a controlled comparison.

\input{revisions/r2o4-1}

\subsubsection{\rone{Variants} Compared} Like our baselines, all methods use a single proxy model, along with the oracle model. We evaluate \method{Task Cascades}, our main approach, as well as \method{Task Cascades (+ Guarantees)}. \input{revisions/r1o1-1} Note that we distinguish \method{Task Cascades}, the approach, from task cascades, the outcome of such an approach; throughout this section, \method{Task Cascades} refers to our approach and is treated as a singular entity (e.g., ``\method{Task Cascades} achieves lower cost''), while a ``task cascade'' refers to an individual cascade configuration discovered by the approach. We also include \rone{different variants of task cascades} to isolate the effect of each component of our approach:

\begin{enumerate}[nosep,leftmargin=*]
\item \textbf{Surrogate Operation Discovery \rone{Variants}.} \method{No Surrogates} disables surrogate discovery, so cascades include only the original user-specified operation at different document fractions. \method{Single-Iteration} generates all surrogate operations in a single batch, without iterative refinement.
\item \textbf{Document Pruning \rone{Variants}.} \method{No Filtering} disables learned document pruning, so all tasks must process full documents. \method{Naive RAG Filter} replaces learned pruning with a simple retrieval filter that ranks chunks by cosine similarity to the embedding of the user-defined operation.
\item \textbf{Alternative Cascade Designs.} Rather than using the default greedy strategy, \method{Selectivity Ordering} constructs the cascade by prioritizing operations with the highest (selectivity~$-1)$/cost ratio, as in prior work on predicate ordering~\cite{hellerstein1993predicate}, where selectivity is the number of documents not classified by the task and thus passed down the cascade. \input{revisions/r2o2} Finally, \method{RAG + NoSur} applies only naive retrieval-based pruning with no surrogate discovery. This variant allows us to quantify the benefit of simple retrieval and isolate the added value of task cascades.
\end{enumerate}

\subsubsection{Metrics and Implementation Details} We report average inference cost in USD (\$) and the fraction of workloads where each method meets the accuracy target $\alpha$. Most experiments use a fixed $\alpha = 0.9$; in one experiment, we vary $\alpha$ from 0.75 to 0.95. For methods with guarantees, we set failure probability $\delta = 0.25$. 

We use OpenAI models: GPT-4o as the oracle model, GPT-4o-mini as the proxy, o1-mini for surrogate generation, and \ttt{text-embedding-\\3-small} for embeddings. We use PyTorch for the relevance classifier in \Cref{sec:restructuring}. All LLM calls use temperature 0. \input{revisions/r1o2-4}

Inference costs are computed using OpenAI API pricing: \$2.50/1M input tokens and \$10.00/1M output tokens for GPT-4o, \$0.15/1M input tokens and \$0.60/1M output tokens for GPT-4o-mini, with a 50\% discount for prefix-cached completions. \input{revisions/r3m5}
All methods were implemented in Python and run on a 2024 MacBook Air (M4 chip) with OpenAI API calls, costing over \$3{,}000 USD.

\subsection{Results}

We now present experimental results addressing the questions outlined above. \Cref{tab:results-90-revised} presents the main cost and accuracy results. \all{We reflect on optimization costs in \Cref{sec:qualitative-insights} and leave a detailed discussion of these costs and latencies to \Cref{sec:optimization-cost-analysis}\cameraready{ in our technical report~\cite{shankar2026task}}.}

\begin{table*}
\centering
\footnotesize
\techreport{\vspace{-10pt}}
\caption{Inference results at a 90\% accuracy target \rthree{on a sample of 1k documents per workload (for datasets with more than 1k documents)}. Each cell shows average accuracy and cost. For all main methods (oracle, 2-model cascades, and task cascades), results are averaged over three trials; variants are single-trial. Baseline costs are shown in USD. Other costs are reported as a multiple of the corresponding 2‑Model Cascade variant (plain or +G). The ``Avg. Cost'' column averages over workloads where the method met the 90\% target. Bolded entries mark the lowest cost among methods on each workload.}
\vspace{-10pt}
\label{tab:results-90-revised}
\resizebox{\textwidth}{!}{%
\begin{tabular}{@{}lccccccccc@{}}
\toprule
\textbf{Method} & \textsc{agnews} & \textsc{court} & \textsc{enron} & \textsc{fever} & \textsc{games} & \textsc{legal} & \textsc{pubmed} & \textsc{wiki\_talk} & \textbf{Avg.\ Cost} \\
\midrule
\methodsmall{Oracle Only}
& \$0.45 & \$11.05 & \$6.34 & \$13.01 & \$3.13 & \$10.20 & \$8.36 & \$2.94 & \$6.94 \\ 
\midrule
\methodsmall{2-Model Cascade}
& 94.1\% | \$0.03
& 89.4\% | \$2.81
& 89.1\% | \$0.61
& 94.4\% | \$0.80
& 89.8\% | \$0.78
& 91.5\% | \$3.23
& 92.5\% | \$0.56
& 91.4\% | \$0.18
& \$1.13 \\
\methodsmall{2-Model Cascades (+G)}
& 94.7\% | \$0.04
& 93.0\% | \$4.06
& 96.2\% | \$1.56
& 96.0\% | \$2.22
& 95.2\% | \$1.29
& 94.8\% | \$4.12
& 93.6\% | \$0.83
& 95.6\% | \$0.49
& \$1.83 \\ 
\midrule
\methodsmall{Task Cascades}
& 95.3\% | 0.66×
& 88.0\% | 0.72×
& 95.5\% | 0.11×
& 90.6\% | 0.54×
& 92.0\% | 1.09×
& 91.3\% | 0.27×
& 91.6\% | 0.84×
& 91.5\% | 0.64×
& 0.59× \\ 
\methodsmall{Task Cascades (+G)}
& 95.4\% | 0.71×
& 90.8\% | {\bf 0.53×}
& 95.6\% | {\bf 0.09×}
& 92.4\% | 0.44×
& 89.2\% | {\bf 0.49×}
& 91.2\% | 0.26×
& 93.5\% | 1.23×
& 92.3\% | {\bf 0.37×}
& {\bf 0.52×} \\
\all{\methodsmall{Task Cascades (Lite)}}
& \all{94.1\% | 0.78×}
& \all{88.6\% | 0.76×}
& \all{95.5\% | 0.14×}
& \all{90.6\% | 0.63×}
& \all{90.2\% | 0.70×}
& \all{90.2\% | 0.41×}
& \all{91.4\% | 0.86×}
& \all{89.9\% | 0.68×}
& \all{0.62×} \\
\midrule
\methodsmall{No Surrogates}
& 95.1\% | 1.43×
& 94.5\% | 1.25×
& 98.0\% | 0.49×
& 92.5\% | 0.86×
& 90.5\% | 1.16×
& 90.0\% | 0.95×
& 91.9\% | 0.79×
& 96.0\% | 2.77×
& 1.21× \\
\methodsmall{Single‑Iteration}
& 92.4\% | {\bf 0.56×}
& 93.8\% | 1.21×
& 95.9\% | 0.15×
& 90.0\% | {\bf 0.30×}
& 89.2\% | 1.21×
& 92.8\% | 1.04×
& 93.4\% | {\bf 0.76×}
& 90.6\% | 0.61×
& 0.66× \\
\methodsmall{No Filtering}
& 93.9\% | 0.74×
& 90.0\% | 1.10×
& 98.5\% | 0.96×
& 95.5\% | 1.82×
& 98.4\% | 1.84×
& 96.1\% | 0.93×
& 95.1\% | 1.47×
& 98.4\% | 3.51×
& 1.55× \\
\methodsmall{Naive RAG Filter}
& 95.1\% | 0.95×
& 90.9\% | 0.94×
& 97.2\% | 0.30×
& 90.9\% | 0.42×
& 85.1\% | 0.61×
& 88.9\% | {\bf 0.12×}
& 90.5\% | 0.88×
& 89.9\% | 0.44×
& 0.65× \\
\methodsmall{Selectivity Ordering}
& 94.2\% | 1.20×
& 90.8\% | 2.33×
& 98.8\% | 5.50×
& 93.1\% | 5.08×
& 90.1\% | 2.32×
& 92.8\% | 1.86×
& 93.8\% | 8.44×
& 96.9\% | 8.81×
& 4.44× \\
\rtwo{\methodsmall{Restructure (Top-25\%)}}
& \rtwo{96.5\% | 2.67×}
& \rtwo{93.9\% | 1.22×}
& \rtwo{99.4\% | 0.85×}
& \rtwo{90.5\% | 0.37×}
& \rtwo{97.1\% | 1.70×}
& \rtwo{95.3\% | 1.78×}
& \rtwo{92.6\% | 2.38×}
& \rtwo{98.1\% | 3.49×}
& \rtwo{1.81×} \\
\methodsmall{RAG + NoSur}
& 96.0\% | 1.83×
& 92.1\% | 1.10×
& 98.2\% | 1.73×
& 91.2\% | 0.54×
& 87.9\% | 1.09×
& 90.0\% | 0.77×
& 90.5\% | 0.94×
& 93.4\% | 1.21×
& 1.16× \\
\bottomrule
\end{tabular}
}
\end{table*}

\begin{figure*}
\centering
\techreport{\vspace{-10pt}}
\includegraphics[width=0.9\linewidth]{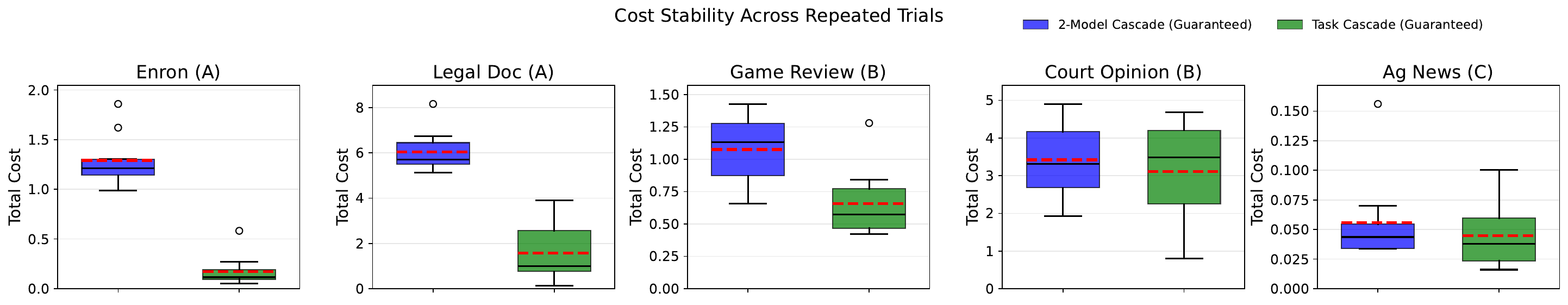}
\vspace{-5pt}
\caption{
\rone{Cost stability of \method{2-Model Cascade (+ Guarantees)} and \method{Task Cascades (+ Guarantees)} across 10 independent runs for five representative workloads.} Each box shows the distribution of total inference cost; red dashed lines indicate the median. \rone{Both methods exhibit high variance.}
}
\label{fig:repeated_trials_cost_stability}
\vspace{-5pt}
\end{figure*}

\begin{figure*}
\centering
\vspace{-10pt}
\includegraphics[width=\linewidth]{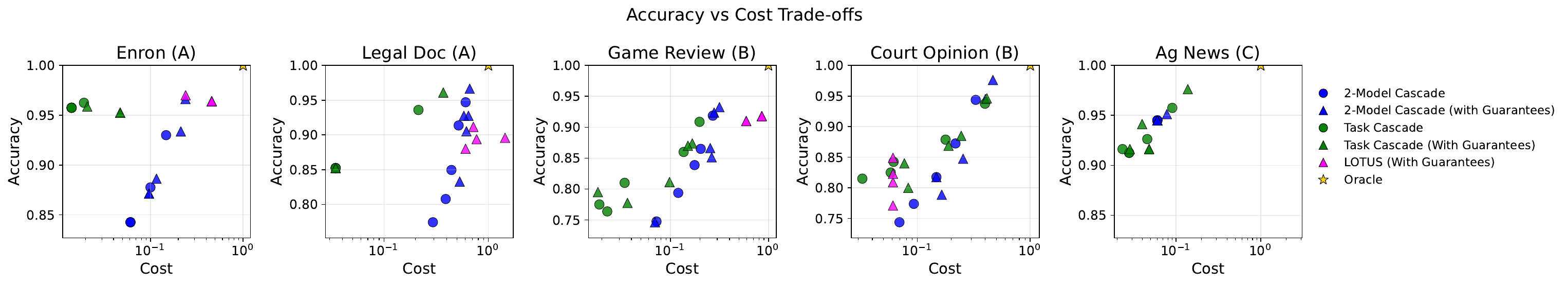}
\vspace{-15pt}
\caption{
Accuracy vs. cost trade-offs as the target accuracy varies 75\% to 95\% (in increments of 5\%). Task cascades (green) dominate the Pareto frontier on easy workloads and provide robust gains or new operating points on harder tasks. On simple workloads like \textsc{agnews}, cost improvements appear mainly at lower targets, where the baselines cannot match. \rtwo{\method{LOTUS (+ Guarantees)} (pink) only supports binary classification workloads, so we exclude \textsc{agnews}.}
}
\vspace{-5pt}
\label{fig:accuracy_vs_cost_tradeoff}
\end{figure*}

\subsubsection{Q1: Cost Savings and Component Importances} 
\label{subsubsec:cost-effectiveness}

\method{Task Cascades} reduces inference cost by 41\% compared to \method{2-Model Cascade}, and 48.5\% when considering accuracy guarantees. \input{revisions/r1o1-2}  Savings are most substantial on workloads with significant irrelevant content (e.g., \textsc{enron} and \textsc{legal}); inference costs drop by up to 6$\times$ relative to \method{2-Model Cascade}.

\topic{Meeting the Accuracy Target} \method{Task Cascades} and \method{Task Cascades (+ Guarantees)} each miss the 90\% accuracy target on one workload by $<1$\%, while \method{2-Model Cascade} fails on 3 of 8 workloads. When accuracy guarantees are required (+ Guarantees variants), both approaches construct cascades using only half the available data, reserving the remainder for threshold adjustment (\Cref{alg:threshold_adjustment}), which increases costs by 1.43$\times$ on average for \method{Task Cascades}.

\topic{Component Contributions} All three core components---greedy task ordering, surrogate operation discovery, document pruning---are essential for robust performance. \method{Selectivity Ordering} performs worst (7.5$\times$ worse than \method{Task Cascades}). \method{Single-Iteration} (1.13$\times$ worse than  \method{Task Cascades}) demonstrates that iterative refinement contributes to cost reduction beyond single-round surrogate generation. 
Surprisingly, both \method{No Surrogates} (1.21$\times$) and \method{No Filtering} (1.55$\times$) are {\em more} expensive than the 2‑Model baseline. Without surrogates, the proxy is limited to the user-defined operation and provides savings on only half the workloads; on remaining tasks, the entire document must be processed, negating cost benefits. Without document pruning, each surrogate must operate on the entire document, making each surrogate operation as expensive as the first stage of \method{2-Model Cascade}. Similarly, \method{Restructure (Top-25\%)} fails to improve cost on 6 of 8 workloads: when the proxy receives too little context to perform the original operation, it cannot confidently resolve documents, escalating more to the expensive oracle.
\input{revisions/r2o2-2} \input{revisions/r1o2-1}

\topic{Comparison with Retrieval-Based Pruning} Our learned document pruning outperforms simple retrieval alternatives. \method{Naive RAG Filter} achieves the second-lowest cost on average, but \method{Task Cascades} reduces cost by an additional 11\%. 
\input{revisions/r3o3}

\topic{Accuracy-Cost Tradeoffs at Fixed Target} 
\input{revisions/r2o3}

\subsubsection{Q2: Accuracy--Cost Tradeoffs Across Targets}
\label{subsubsec:acrosstargets}
To evaluate performance across different accuracy requirements, we vary the target accuracy from 75\% to 95\% in 5\% increments, using the same workload groups as in \Cref{subsubsec:stability}. \Cref{fig:accuracy_vs_cost_tradeoff} shows the resulting accuracy--cost tradeoff. \input{revisions/r2o4-2} On easier Group A tasks (\textsc{enron}, \textsc{legal}), \method{Task Cascades} consistently dominates the Pareto frontier, matching or exceeding model cascade accuracy at much lower cost for any target. For more challenging Group B workloads (\textsc{games}, \textsc{court}), \method{Task Cascades} provides the greatest cost savings at lower target accuracies (75\%---85\%), because the agent can choose from a larger set of surrogate operations that meet these relaxed accuracy targets. As the target accuracy increases, fewer surrogates will satisfy the stricter requirements, so cost advantages diminish---but \method{Task Cascades} still often matches or outperforms the baseline, perhaps due to document pruning. On Group C (\textsc{agnews}), where the proxy already achieves high accuracy, there is little room for improvement: \method{Task Cascades} performs similarly to the baseline at high targets but can identify lower-cost solutions at lower accuracy targets.

\subsubsection{Q3: Consistency and Variance Across Repeated Trials}
\label{subsubsec:stability}

We next evaluate the consistency of task cascade outcomes across repeated runs, given the inherent stochasticity of LLM-based surrogate discovery. 
To provide a representative analysis, while managing compute and cost constraints, we select five workloads spanning a range of document lengths and task complexities. 
We organize these into three groups based on observed outcomes in \Cref{tab:results-90-revised}: Group A (\textsc{enron} and \textsc{legal}), where task cascades deliver the largest cost reductions; Group B (\textsc{games} and \textsc{court}), where it is most difficult to find cheap and accurate cascades; and Group C (\textsc{agnews}), where the proxy model alone is already highly accurate and documents are short, impeding savings from document pruning techniques.

\input{revisions/r1o3}

\input{revisions/opcostmain}

\subsection{Insights from Analyzing Task Cascades}
\label{sec:qualitative-insights}

\topic{Choosing Variants of the Task Cascade Approach} \input{revisions/r1o2-2}

\topic{Parameter Selection Guidance} \input{revisions/r1-o2-3}

%% file: revisions/r3m4.tex
\rthree{The original task in~\cite{liu2024declarative} involved fraud detection, but fraud cases are extremely rare in the dataset, so a trivial ``always predict no fraud'' strategy achieves very high accuracy. To create a more balanced classification problem suitable for evaluating task cascades, we adapt the original fraud detection task to identifying emails sent by C-suite or VP-level executives, which provides a more even class distribution}.

%% file: revisions/r3m2.tex
\rthree{We sampled datasets to stay within a \$5,000 API budget; running all variants described below across full datasets with multiple accuracy targets and trials would have cost hundreds of thousands of dollars.}

%% file: revisions/r2o4-1.tex
\rtwo{For \textbf{Q2}, which evaluates cost–accuracy tradeoffs across varying accuracy targets, we additionally include \method{LOTUS (+ Guarantees)}~\cite{patel2024lotus}. LOTUS embeds model cascades inside filter operators but provides guarantees on \emph{precision} and \emph{recall}, rather than overall accuracy. Because LOTUS requires separate precision and recall targets instead of a single accuracy target, we cannot directly compare it to our approach at fixed accuracy levels (as in \textbf{Q1} and \textbf{Q3}). For \textbf{Q2}, we implement the LOTUS baseline using their \ttt{sem\_filter} operator and evaluate it across all combinations of precision and recall targets in $\{0.60, 0.80, 0.90, 0.95\} \times \{0.60, 0.80, 0.90, 0.95\}$ on each binary classification workload.}

%% file: revisions/r1o1-1.tex
\all{We also evaluate \method{Task Cascades (Lite)}, a lightweight variant that minimizes optimization cost by using the same parameters as \method{Task Cascades} but evaluating all candidate surrogate tasks only with the proxy model, never the oracle, substantially reducing $C_{\text{eval}}$ (as described in \Cref{sec:cost-model}).}

%% file: revisions/r2o2.tex
\rtwo{\method{Restructure (Top-25\%)} applies our learned document restructuring to reorder each document, then allows the proxy to process only the top 25\% of the restructured content, while the oracle operates on the full document as usual. This approach isolates the effect of learned restructuring and pruning with a 2-task cascade, without introducing surrogate discovery or having more than 2 tasks.}

%% file: revisions/r1o2-4.tex
\rone{We consider four document fractions $\mathcal{F} = \{0.1, 0.25, 0.5, 1.0\}$ and run surrogate discovery for three iterations ($n_a = 3$) with three surrogate operations per iteration ($n_s = 5$). These parameter values were held constant across all workloads. While more iterations or surrogates may yield cheaper cascades at higher offline cost, we find these defaults work well across diverse tasks, and run more experiments in \Cref{sec:param-sensitivity}\cameraready{ in our technical report~\cite{shankar2026task}} to demonstrate that task cascades are robust to parameter choices.}

%% file: revisions/r3m5.tex
\rthree{Embedding costs using (\$0.02/1M tokens) are also included in all reported inference costs; these represent about 0.13$\times$ the cost of GPT-4o-mini inference and are negligible compared to overall LLM inference costs.}

%% file: revisions/r1o1-2.tex
\all{\method{Task Cascades (Lite)}, our lightweight variant with reduced optimization cost, achieves nearly identical savings (0.62$\times$ vs. 0.59$\times$ on average) while requiring 75\% less optimization investment.}

%% file: revisions/r2o2-2.tex
\rtwo{Similarly, \method{Restructure (Top-25\%)} fails to improve cost on 6 of 8 workloads: when the proxy receives too little context to perform the original operation, it cannot confidently resolve documents, escalating more to the expensive oracle.} 

%% file: revisions/r1o2-1.tex
\rone{Across workloads, simpler variants can fail in significant ways: \method{No Filtering} is up to 8$\times$ worse than \method{Task Cascades}, \method{No Surrogates} up to 4$\times$ worse, and the maximum degradation across all variants averages 12$\times$ worse than \method{Task Cascades}. By combining multiple techniques, \method{Task Cascades} achieves robust cost savings across diverse workloads.}

%% file: revisions/r3o3.tex
\rthree{More notably, \method{RAG + NoSur} (retrieval-based pruning without surrogate discovery) costs 1.16$\times$ more than \method{2-Model Cascade} on average: naive document reordering can make restructured documents harder for the proxy to process, escalating more documents to the expensive oracle. Adding surrogate discovery (\method{Naive RAG Filter}) addresses this degradation, making it cheaper than \method{2-Model Cascade} on all workloads.}

%% file: revisions/r2o3.tex
\rtwo{At the fixed 90\% accuracy target, task cascades explore a broader configuration space, enabling more precise cost-accuracy tradeoffs. In some cases, task cascades identify plans that are both cheaper {\em and} more accurate than baselines: on \textsc{enron}, \method{Task Cascades} achieves 95.5\% accuracy at 0.11$\times$ cost vs. \method{2-Model Cascade}'s 89.1\% accuracy at 1.0$\times$ cost. In other cases, task cascades find cheaper plans closer to the target: on \textsc{fever}, \method{Task Cascades} achieves 90.6\% accuracy at 0.54$\times$ cost vs. \method{2-Model Cascade}'s 94.4\% accuracy at 1.0$\times$ cost. Overall, \method{Task Cascades} can search over more configurations and, many times, land closer to the target accuracy.}

%% file: revisions/r2o4-2.tex
\rtwo{We include \method{LOTUS (+ Guarantees)} on binary classification workloads (all except \textsc{agnews}), using the setup described in \Cref{sec:setup}.}
\rtwo{Overall, \method{Task Cascades} reduces inference costs and extends the Pareto frontier across both easy and hard tasks, compared to all baselines. \method{LOTUS (+ Guarantees)} performs comparably to \method{2-Model Cascade (+ Guarantees)} on all workloads.}

%% file: revisions/r1o3.tex
We next evaluate the consistency of task cascade outcomes across repeated runs, given the inherent stochasticity of LLM-based surrogate discovery. 
To provide a representative analysis while managing compute and cost constraints, we select five workloads spanning a range of document lengths and task complexities. 
\rone{We group them by observed difficulty in \Cref{tab:results-90-revised}:
Group A (\textsc{enron}, \textsc{legal}) where task cascades achieve the largest cost reductions;
Group B (\textsc{games}, \textsc{court}) where it is hardest to find cheap, accurate cascades;
and Group C (\textsc{agnews}) where documents are short and the proxy model is already highly accurate. 
As shown in \Cref{fig:repeated_trials_cost_stability}, we evaluate both \method{2-Model Cascade (+ Guarantees)} and \method{Task Cascades (+ Guarantees)} across ten independent runs per workload. On Group A workloads, \method{Task Cascades (+ Guarantees)} reduces average cost by 86.5\% on \textsc{enron} and 74.1\% on \textsc{legal}, and even the {\em most expensive} task cascade run remains {\em cheaper than the cheapest} 2-model cascade.
On Group B workloads, both methods show wider variance due to increased task difficulty, but the {\em 75th percentile} task cascade cost is still {\em lower than the 25th percentile} of the baseline cost. 
On Group C (\textsc{agnews}), where the baseline is already near-optimal, \method{Task Cascades (+ Guarantees)} match its cost and variance.}

\rone{Both methods exhibit notable run-to-run variance. 
\method{2-Model Cascade (+ Guarantees)} also shows high variability---likely because the small development set size and data-splitting procedure for threshold adjustment introduce uncertainty in cascade construction.
\method{Task Cascades (+ Guarantees)} consistently achieves lower mean and median cost across all workloads, demonstrating that variability does not undermine its advantage.}

%% file: revisions/opcostmain.tex
\subsubsection{Q4: Optimization Costs}
\label{subsubsec:optimization-costs}
\all{So far, we have focused on inference costs, as these dominate in production-scale deployments, where offline optimization costs are offset by online gains. To examine when this investment ``pays off'', we report optimization costs, break-even points (when total cost becomes lower than running \method{Oracle Only}), and estimated costs for processing 1 million documents (selected to showcase the impact at scale). Full details are provided in \Cref{sec:optimization-cost-analysis}, described in \Cref{tab:optimization-and-breakeven}\cameraready{ in our technical report~\cite{shankar2026task}}.}

\all{At 1M documents, optimization costs become negligible and task cascades provide substantial savings. \method{Task Cascades (Lite)}---our lightweight variant, described in \Cref{sec:setup}---provides cost savings over \method{2-Model Cascade} on all workloads, with average savings of 37.5\% and up to 86\% on \textsc{enron}. \method{Task Cascades} provides savings on 7 of 8 workloads (average 36\%, up to 87\% on \textsc{enron}); on \textsc{games}, where \method{Task Cascades} does not provide savings, \method{2-Model Cascade} fails to meet the 90\% accuracy target (\Cref{tab:results-90-revised}).}

\all{These savings are achieved even though \method{Task Cascades} incurs higher optimization costs than \method{2-Model Cascade} (7.9$\times$--27.0$\times$ across workloads, mean 11.4$\times$) and takes longer to build ($\sim14\times$ optimization latency). Our lightweight \method{Task Cascades (Lite)} variant costs 2.2$\times$--4.7$\times$ the \method{2-Model Cascade} baseline (mean 2.9$\times$) while preserving 95\% of inference savings. Optimization cost is recovered after processing 2{,}986 documents on average for \method{Task Cascades}, while \method{Task Cascades (Lite)} requires only 678---modest thresholds relative to our workload sizes (median 129k documents, as in \Cref{tab:dataset-stats}). Inference latencies are within the same order of magnitude (\Cref{tab:latency} in \Cref{sec:optimization-cost-analysis}\cameraready{ in our technical report~\cite{shankar2026task}}).}

\all{In summary, while task cascades require higher upfront optimization investment, this cost quickly becomes negligible at realistic scales, demonstrating that task cascades are practical and cost-effective for production deployments.}

%% file: revisions/r1o2-2.tex
\rone{The \method{Task Cascades} approach provides the most consistent cost reductions across workloads, making it a reliable default. Simpler variants may be preferable in specific scenarios: \method{Naive RAG Filter} when engineering effort is limited and documents have clearly identifiable relevant sections; \method{No Surrogates} when the original operation is not very difficult and the proxy can accurately perform it; and variants without filtering when documents are very short (e.g., \textsc{agnews}). The \method{Single-Iteration} variant is particularly promising, achieving strong performance with reduced offline cost and likely to improve as LLM agents improve in their ability to reason and handle complex tasks.  While individual variants occasionally outperform the full approach, \method{Task Cascades} is 20\% cheaper than the next best variant on average, making it the recommended starting point.}

%% file: revisions/r1-o2-3.tex
\rone{Task cascades depend on three hyperparameters: the number of surrogate operations per iteration ($n_s$), refinement iterations ($n_a$), and document fraction sets ($\mathcal{F}$). We conduct sensitivity analysis on two representative workloads, varying $n_s \in \{3, 5, 10\}$, $n_a \in \{1, 2, 3\}$, and testing different $\mathcal{F}$ configurations (detailed results in \Cref{sec:param-sensitivity}\cameraready{ in our technical report~\cite{shankar2026task}}). Task cascades consistently outperform baselines across all configurations, demonstrating robustness to parameter choices. Based on these experiments, we recommend: $n_s \geq 3$, $n_a \geq 1$, and document fractions spanning from small (0.1 or 0.25) to full (1.0). We find that $n_s$ has a stronger effect on cost reduction than $n_a$, so practitioners should prioritize exploring diverse surrogates over extensive refinement.}

%% file: sections/related.tex
\section{Related Work}
\label{sec:related}

We cover three areas: LLM-powered data processing, cost-efficient LLM execution, and query rewriting with LLMs.

\topic{LLM-Powered Data Processing} Modern data systems increasingly integrate LLMs as first-class operators for natural language extraction and transformation. Industrial systems like Databricks, DuckDB, Snowflake, and Google's AlloyDB support LLM-powered {\em filter} and {\em map} functions in SQL~\cite{databricks-llm, duckdb-llm, alloydb-llm, snowflake-llm}. LOTUS~\cite{patel2024lotus}, ThalamusDB~\cite{jo2024thalamusdb}, and Aryn~\cite{anderson2024design} provide pandas-based, SQL, and Spark-like interfaces respectively. Palimpzest~\cite{liu2024declarative}, DocETL~\cite{shankar2024docetl}, and AOP~\cite{wang2025aop} offer custom DSLs for LLM-powered data processing. However, these systems typically invoke LLMs on each document independently, resulting in high inference costs. 
Some systems focus on supporting LLM-powered data processing for particular query or data modalities~\cite{caesura, lin2025querying, arora2023language, lu2025vectraflow}. 

\topic{Cost Optimization for LLM Inference} Strategies for reducing LLM inference costs include cost-based optimizers, specialized models, caching, ensembling, and model cascades. ABACUS~\cite{russo2025abacus} introduces a Cascades-style optimizer~\cite{Graefe1995TheCF} over different ``physical'' implementations of common LLM-powered operators, while ELEET~\cite{urban2024eleet} replaces LLM operators with trainable models. Other work reorders LLM calls for KV-cache reuse~\cite{liu2025optimizing}. Some methods profile multiple LLMs and aggregate their outputs to reduce cost. For example, ThriftLLM~\cite{huang2025thriftllm} and LLMBlender~\cite{jiang2023llm}, instead of using only the output of the last processed stage as in model cascades, ensemble predictions from all evaluated models. However, in the primary setting we study---where cascade models differ substantially in cost and quality---there is no additional value to be gained by ensembling the oracle's output (when available) with that of the proxy. SpareLLM~\cite{jo2024smart} profiles a number of LLMs for a given family of tasks, but ultimately selects a single model to handle the remaining tasks, which in our case boils down to a cascade with a single step, either a proxy model or an oracle, for all items, independent of confidence.
Our work builds on model cascades~\cite{viola2001rapid, kang2017noscope, lu2018accelerating, kang13blazeit, anderson2019physical, cao2022figo, kang2020approximate, chenfrugalgpt, patel2024lotus, yue2024large}, generalizing them by varying not only the model but also the operation and document fraction at each stage.

\topic{Query Rewriting with LLMs} LLM-powered query rewriting has proven effective in retrieval, question-answering~\cite{ma2023query, peng2024large, balaka2025pneuma, chen2025can}, and traditional query optimization~\cite{dharwada2025query, song2025quite, zhou2023learned, li2024llm}. DocETL~\cite{shankar2024docetl} decomposes complex tasks into logically equivalent subtasks for accuracy. In contrast, we rewrite tasks for cheaper, possibly-incorrect operations, dramatically expanding the rewrite space. Task cascades may not resemble logical decompositions---the same operation can appear multiple times at different document fractions.

Our work also builds on approximate query processing (AQP)~\cite{garofalakis2001approximate,chaudhuri2017approximate}. While tailored for unstructured data and LLMs, our components mirror classical AQP: surrogate operations resemble approximate predicates~\cite{ilprints351}, document fractions leverage sampling, and we use concentration bounds for accuracy guarantees~\cite{chaudhuri1998random, kang2020approximate}.

%% file: sections/conclusion.tex
\section{Conclusion}

We introduce {\em task cascades}, a generalization of the model cascades framework for efficient LLM-powered unstructured text processing. Task cascades vary the model, operation, and document fraction at each stage to minimize inference cost while meeting accuracy targets. 
Across eight real-world workloads, task cascades reduce inference cost by 48.5\% on average compared to model cascades (with comparable guarantees) and by 86.2\% relative to oracle-only inference, with ablations confirming the importance of all components of our approach. Overall, task cascades provide a practical and extensible solution for scalable unstructured data analysis.

\begin{acks}
We acknowledge support from grants DGE-2243822, IIS-2129008, IIS-1940759, and IIS-1940757 awarded by the National Science Foundation, funds from the State of California, an NDSEG fellowship, funds from the Alfred P. Sloan Foundation, as well as EPIC lab sponsors: Adobe, Google, G-Research, Microsoft, PromptQL, Sigma Computing, Bridgewater, and Snowflake. Compute creditswere provided by Azure, Modal, NSF (via NAIRR), and OpenAI.
\end{acks}

%% file: sections/app.tex
\section{Cascade Assembly and Statistical Guarantees for Cascade Accuracy}
\label{app:guarantees}

\all{\subsection{Greedy Cascade Assembly Algorithm}
\label{app:greedy-cascade-assembly}}

\all{In this section, we provide the pseudocode for the greedy cascade assembly algorithm referenced in \Cref{ssec:greedy-assembly}. This algorithm takes as input a set of eligible tasks (filtered and assigned thresholds via \Cref{alg:find-thresholds}) and greedily constructs an ordered cascade by iteratively selecting the task that maximally reduces inference cost while maintaining per-task accuracy requirements.}

\begin{algorithm}[h]
\footnotesize
\SetAlgoLined
\KwIn{Eligible tasks $\mathcal{T}_{\mathrm{eligible}}$ with thresholds, development set $D_{\mathrm{dev}}$, accuracy target $\alpha$}
\KwOut{Ordered cascade $\pi$}
\BlankLine
$\pi \gets \emptyset$ \;
\While{true}{
    $\texttt{best\_task} \gets$ None;\quad $\texttt{best\_cost} \gets$ cost of $\pi$\;
    \ForEach{unused task $T$ in $\mathcal{T}_{\mathrm{eligible}}$}{
        $\pi' \gets \pi + [T]$ \tcp*{Append $T$ to end of cascade}
        Execute $\pi'$ on $D_{\mathrm{dev}}$\;
        \If{for every task in $\pi'$, accuracy on the documents it classifies (i.e., with confidence above threshold) $\geq \alpha$ and cost$(\pi') < \texttt{best\_cost}$}{
    $\texttt{best\_task} \gets T$;\quad $\texttt{best\_cost} \gets$ cost$(\pi')$\;
}
    }
    \If{$\texttt{best\_task}$ is None}{
        \textbf{break}
    }
    $\pi \gets \pi + [\texttt{best\_task}]$ \tcp*{Update current cascade}
}
\Return $\pi$\;
\caption{Greedy Cascade Assembly (detailed)}
\label{alg:greedy-cascade}
\end{algorithm}

\subsection{Estimator Details and Proof of \Cref{thm:guarantees}}
\label{app:guarantees-proof}

\textbf{Estimation function}. To meet the target accuracy, the function $\mathcal{E}$ needs to correctly estimate, based on the sampled set $D_V$, whether a threshold set meets the accuracy target or not. More formally, let $\pi_i$ be the model cascade when using the threshold set $\mathbf{t}_i$, and let $\pi(x_j)$ for $(x_j, y_i)\in D_V$ be the output of the model cascade on input sample $x$ whose ground-truth answer is $y$. Define $X_{j}^i=\mathbb{I}[\pi_i(x_j)=y_j]$.  Note that $\mathbb{E}(X_{j}^i=1)=\mathrm{Acc}_D(\pi_i)$, so that the goal of $\mathcal{E}$ is to estimate, using a set of observed i.i.d Bernoulli random variables $X^i=\{X_{j}^i; \forall j\}$ if their mean is more than the target $T$ or not. This can be achieved by an application of common concentration bounds such as Hoeffding's inequality. We use the recent results by \citet{waudby2024estimating} that provides tighter bounds that Hoeffding's inequality as shown in \cite{waudby2024estimating}. The results of \cite{waudby2024estimating} takes both means and \textit{variances} of the observations into account to estimate their mean, thus yielding tighter bounds. The following is a restatement of the \Cref{thm:guarantees} by \cite{waudby2024estimating}, defining our function $\mathcal{E}$ and showing that it can estimate whether the mean of observation is more than $T$ or not with high probability.

\begin{lemma}[Corollary to \Cref{thm:guarantees} by \cite{waudby2024estimating}]\label{lemma:estimation}
    Consider the $i$-th cascade threshold set $\mathbf{t}_i$ and the corresponding Bernoulli random variables $X^i$ with $|X^i|=k$. If $\mathrm{Acc}_D(\mathbf{t}_i)<T$ and for a confidence parameter $\alpha\in[0, 1]$,
    \begin{align}\label{eq:est_prism_bound}
        \mathbb{P}(\mathcal{E}(\mathbf{t}, D_V)=True)\leq \alpha, \quad \text{where}
    \end{align}
    \begin{align}\label{eq:est_prism_def}
        \mathcal{E}(\mathbf{t}_i, D_V)=\mathbb{I}[\exists j\in[k]\,\text{s.t.}\,\mathcal{K}(T, X^i[:j])\geq\frac{1}{\alpha}].
    \end{align}

    $\mathcal{K}(T, X)$ for a set of $i$ random variables $X=\{X_1, ..., X_i\}$ is defined as 
    \begin{align}\label{eq:k_in_precision}
        \mathcal{K}(T, X)=\Pi_{j=1}^{i}(1+\min(\lambda_j, \frac{3}{4T})\times(X_j-T)), 
    \end{align}
    \begin{align*}
        \hspace{-13pt}\lambda_i=\sqrt{\frac{2\log(2/\delta)}{i\log(i+1)\hat{\sigma}^2_{i-1}}},\; \hat{\sigma}_i^2=\frac{1/4+\sum_{j=1}^i(X_j-\hat{\mu}_j)^2}{i+1}, \;\hat{\mu}_i = \frac{1/2+\sum_{j=1}^iX_j}{i+1}.
    \end{align*}
\end{lemma}

\textbf{Proof of \Cref{thm:guarantees}}.  Lemma~\ref{lemma:estimation} shows that if the threshold set at the $i$-th iteration does not meet the target, then the probability that the estimation function returns True is bounded by $\alpha$. We use this result to show the total probability that \Cref{alg:threshold_adjustment} returns a threshold that doesn't meet the target is also bounded by $\alpha$. To do so, let $\mathbf{t}_{i^*}$ be the threshold set with the lowest $i^*$ that does not meet the target. Note that  \Cref{alg:threshold_adjustment} returns a threshold up to the first threshold that it estimates does not meet the target. Thus, if $\mathcal{E}(\mathbf{t}_{i^*}, D_V)=False$, then \Cref{alg:threshold_adjustment} returns a threshold set that meets the target accuracy (or returns not found).  As such, \Cref{alg:threshold_adjustment} returns a threshold that fails to meet the target only if $\mathcal{E}(\mathbf{t}_{i^*}, D_V)=True$. The probability of the latter is bounded by $\alpha$ according to \Cref{lemma:estimation}, so the probability of threshold selected by \Cref{alg:threshold_adjustment} not meeting the target is also bounded by $\alpha$.\qed

\subsection{Threshold Adjustment Procedure}
\label{app:guarantees-adjustment}

\begin{algorithm}
\footnotesize
\SetAlgoLined
\KwIn{Candidate threshold lists $\mathbf{p}_T^c$ for each task and class; validation set $D_V$; estimator $\mathcal{E}$ (constructed to certify accuracy $\alpha$ at failure probability $\delta$)}
\KwOut{Adjusted thresholds $\mathbf{t}^*$ such that $\Pr[\mathrm{Acc}(\pi^*) < \alpha] \leq \delta$}
\BlankLine
\texttt{best\_t} $\gets$ thresholds with maximal shift (most conservative)\;
\For{$s = s_{\max}$ \textbf{down to} $0$}{
\ForEach{class $c$ for each task $T$}{
Initialize $\mathbf{t} \gets \varnothing$\;
\If{$s <$ length of $\mathbf{p}_T^c$}{
    Set $\tau^c = p_s$ \tcp*{$\mathbf{p}_T^c = (p_0, p_1, \ldots, p_{s_{\max}})$}
}

\Else{
Set $\tau^c = \infty$ \tcp*{Disable class if shift exceeds available values}
}
Add $\tau^c$ to $\mathbf{t}$ for $(T, c)$\;
}
\If{$\mathcal{E}(\mathbf{t}, D_V)$ is \texttt{True}}{
        \texttt{best\_t} $\gets \mathbf{t}$\;
    }
\Else{
\textbf{break}
}
}
\Return \texttt{best\_t};
\caption{ThresholdShift}
\label{alg:global_threshold_shift_app}
\end{algorithm}

The final step of cascade construction (see \Cref{ssec:guarantees}) adjusts class thresholds to ensure the cascade achieves the target accuracy $\alpha$ on unseen data with probability at least $1 - \delta$. For each class $c$ in each task, we construct a shift list of candidate thresholds $p_1 < p_2 < \cdots < p_{k}$, representing increasingly conservative cutoffs above the original threshold $\tau^c$. For each shift value $s$ from $k$ down to $0$, we set the threshold for class $c$ to $p_s$ for $s > 0$, and to $\tau^c$ for $s = 0$; if no such threshold exists for $s$, we set $\tau^c = \infty$, disabling predictions of that class at that task. At each candidate threshold set $\mathbf{t}_s$, we evaluate the cascade on the validation set $D_V$ and apply the estimator $\mathcal{E}$ to certify whether the accuracy guarantee is satisfied. We return the least conservative (smallest $s$) threshold set that meets the guarantee. If no threshold set passes, we revert to the oracle-only cascade (this was never triggered in our experiments).

The full algorithm is presented in \Cref{alg:global_threshold_shift_app}.

\begin{algorithm}
\footnotesize
\SetAlgoLined
\KwIn{Original operation $o_{\mathrm{orig}}$, development set $D_{\mathrm{dev}}$, models $\mathcal{M}$, document fractions $\mathcal{F}$, accuracy target $\alpha$, failure probability $\delta$, rounds $n_a$, surrogates per round $n_s$}
\KwOut{Final cascade $\pi^*$}
\BlankLine

\tcp{Initialize candidate tasks}
$\mathcal{T} \gets \{(m, o_{\mathrm{orig}}, f) : m \in \mathcal{M}, f \in \mathcal{F}\}$

\For{$r = 1$ {\bf to} $n_a$}{
    \tcp{Assemble best cascade using CascadeAssembly (\Cref{alg:greedy-cascade})}
    $\pi_{\text{current}} \gets \textsc{CascadeAssembly}(\mathcal{T}, D_{\mathrm{dev}}, \alpha, \delta)$

    \tcp{Collect feedback and failure cases}
    \ForEach{task $T$ in $\mathcal{T}$}{
        Determine if $T$ is selected in $\pi_{\text{current}}$; record its coverage (number of documents classified) and up to 10 hard misclassified examples near threshold (confidence just above $\tau$).\\
    }
    Identify documents not classified by any non-oracle task (i.e., routed to the oracle); extract minimal supporting spans using the oracle for prompt brevity.
    
    \tcp{Elicit new surrogates from agent}
    Provide agent with: (i) user operation, (ii) oracle-extracted failure cases, (iii) per-task statistics, (iv) explicit surrogate generation instructions.\\
    Receive $n_s$ new surrogate operations.

    \tcp{Expand candidate set}
    \ForEach{new surrogate $o_{\text{sur}}$}{
        Add $(m, o_{\text{sur}}, f)$ to $\mathcal{T}$ for all $m \in \mathcal{M}, f \in \mathcal{F}$
    }
}
\Return{$\textsc{CascadeAssembly}(\mathcal{T}, D_{\mathrm{dev}}, \alpha, \delta)$}
\caption{Agentic Loop for Surrogate Cascade Refinement}
\label{alg:agentic_loop_surrogate}
\end{algorithm}

\section{Surrogate Operation Generation: Agentic Loop}
\label{app:surrogate-operation-gen}

\Cref{alg:agentic_loop_surrogate} describes the procedure for surrogate operation generation. At each iteration, we provide the LLM agent with a prompt structured as in \Cref{fig:prompt-template-surrogate agent}.

\begin{figure*}[h]
\centering
\begin{lstlisting}
Your job is to propose [n_s] simple surrogate operations for the classification task below. Each surrogate must target a different detection type from the following list:
- Entity Detection (checks for presence of a specific entity)
- Event Detection (detects a particular event or outcome)
- Relationship Detection (identifies a connection or association)
- Context Detection (determines the broader setting)
- Attribute Detection (checks for a property or attribute)
- Any other type not mentioned above

Each surrogate should be much simpler than the original task, and you must use a unique detection type for each. For a classification task, each surrogate's outputs must be a subset of the original task's outputs (if multiclass, may also output -1 for ``none of the above''; if binary, output must be True or False).

TASK:
[insert user-defined operation here]

FAILURE CASES:
Here are examples of documents the current cascade fails to classify with any non-oracle task. Only the most relevant text span from each document is shown (as extracted by the oracle model):

[insert minimal relevant spans]

TASK STATISTICS:
For each candidate task:
- Selected: [yes/no]
- Coverage: [number of documents classified]
- Hard examples: Up to 10 misclassified documents for this task, showing only the minimal relevant span from each document (as extracted by the oracle model).

INSTRUCTIONS:
Propose [n_s] surrogate operations, each corresponding to a different detection type above and distinct from surrogates previously generated. Surrogates should be iteratively refined based on the task statistics and failure patterns above.

For each surrogate, provide:

PROMPT: <a concise classification instruction, with allowable outputs matching the original task>
RATIONALE: <what it detects, which detection type it uses, and why it is simpler or complementary to previous surrogates>

EXAMPLE  (for a binary task; determining if a medical article describes adverse drug reactions):

TASK:
Does this article describe a negative reaction to a drug? Output True if a negative drug reaction is present, False if not.

Example surrogate operations:

PROMPT: Does the article describe a patient outcome such as rash, nausea, or toxicity? If yes, output True. Otherwise, output False.
RATIONALE: Event Detection: identifies negative outcomes, without requiring explicit attribution to a drug.

PROMPT: Is the article a case report? If yes, output True. Otherwise, output False.
RATIONALE: Context Detection: case reports often highlight adverse reactions in individual patients.

YOUR ANSWER HERE:

\end{lstlisting}
\caption{Prompt template used to elicit surrogates from the agent.}
\label{fig:prompt-template-surrogate agent}
\end{figure*}

The prompt in \Cref{fig:prompt-template-surrogate agent} is used at every iteration of the agentic loop, updated each time with new failure cases and task statistics. To help the agent generate useful surrogate operations, we provide a short list of concrete detection strategies, such as entity detection, event detection, relationship extraction, and context cues, which are easily translated into simple classification rules. We also include an in-context example showing how to turn these strategies into specific instructions. This example is only for demonstration and is not used in our actual experiments. The complete prompt template is available in our codebase.

While these detection strategies help guide the agent toward productive proposals, they are not strictly necessary---a good agent may discover effective surrogates even without this guidance. A systematic study of prompt engineering for surrogate generation, including the potential for agent fine-tuning, is left to future work.

\begin{figure}
    \centering
    \includegraphics[width=\linewidth]{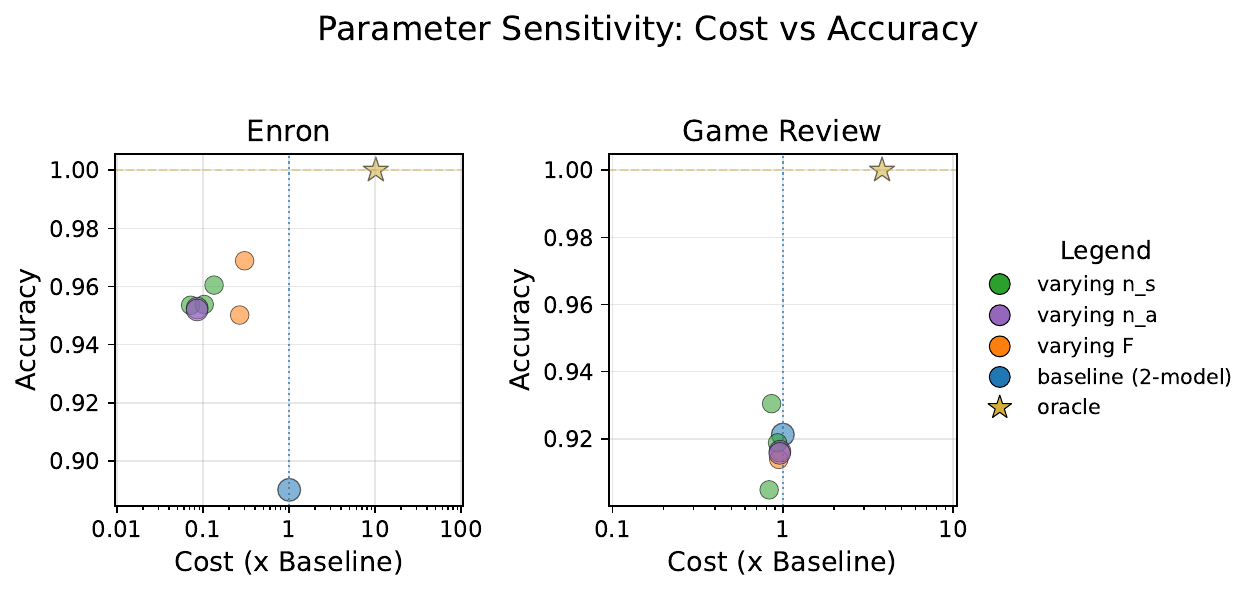}
    \caption{\rone{Parameter sensitivity analysis on two workloads representing different performance regimes. Each point shows cost (relative to \method{2-Model Cascade} baseline) and accuracy for different parameter settings. \method{Task cascades} demonstrates low variance across all parameter configurations, even on \textsc{games} where baseline performance is strong.}}
    \label{fig:param-sensitivity}
\end{figure}

\section{Examples of Surrogate Operations}
\label{app:examples-surr} 

\begin{figure}
    \centering
    \includegraphics[width=0.9\linewidth]{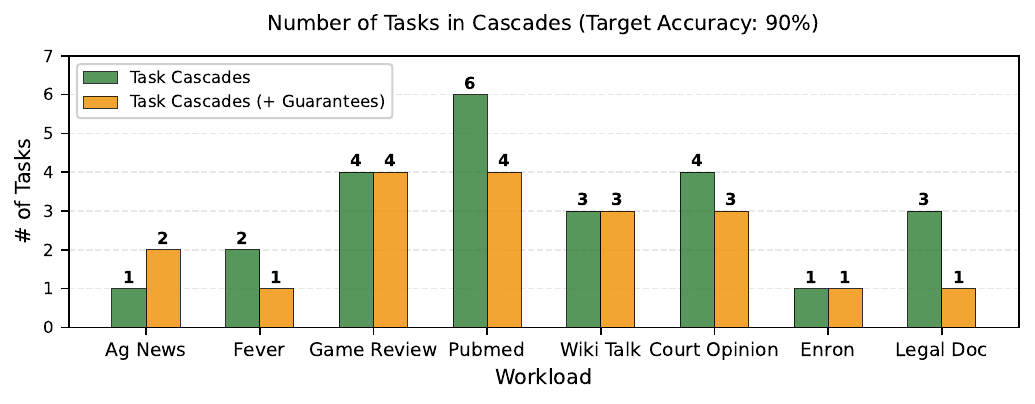}
    \vspace{-5pt}
    \caption{Number of tasks in each cascade for the \method{Task Cascades} and \method{Task Cascades (+ Guarantees)} methods at $\alpha=90\%$.}
    \label{fig:num-tasks-cascade}
    \vspace{-10pt}
\end{figure}

To understand how task cascades achieve cost savings, we examined the surrogate instructions produced by our agent and identified four main types. {\bf Keyword or phrase-based surrogates} rely on the presence of predictive terms or domain jargon; for example, in \textsc{wiki\_talk}, a surrogate checked for ``WP:3RR'' or ``three-revert rule,'' and in \textsc{games}, for comparative phrases like ``better than'' or ``prefer over.'' {\bf Class-specific surrogates} target features associated with a single class: for \textsc{agnews}, one surrogate detected references to corporate mergers or financial metrics to identify business articles; for \textsc{pubmed}, another flagged documents as ``research'' if they contained phrases like ``in vitro'' or ``in vivo.'' {\bf Semantic pattern-based surrogates} capture higher-level cues correlated with the target; for example, in \textsc{enron}, a surrogate looked for a formal signature with title and contact information to identify executive emails. Fourth, {\bf sequential decomposition surrogates} break the original operation into simpler steps; in \textsc{games}, one surrogate checked only whether any other game was mentioned---a necessary prerequisite for assessing comparative sentiment. In some cases, the agent found no helpful surrogate; for example, in \textsc{fever}, the best cascade simply reapplied the original user instruction to a subset of the document. The agent also sometimes reused the same surrogate at multiple document fractions, such as in \textsc{legal}, where it checked for ``agreement not to sue or challenge intellectual property rights'' at both 0.1 and 1.0 fractions.

\rone{\section{Parameter Sensitivity Analysis}
\label{sec:param-sensitivity}}

\input{revisions/paramsensitivity}

\input{revisions/opcosttable}

\input{revisions/opcost}

\section{Workload Prompt Templates}
\label{app:task-prompts}

This section lists the exact prompt templates used for each classification workload in our experiments (\Cref{sec:evaluation}). Workloads were selected as described in \Cref{sec:evaluation}.

To ensure clear and consistent task specification, we used Claude 3.5 Sonnet to generate each instruction, explicitly prompting it to ``make this an unambiguous classification task for an LLM.'' Each template includes a \texttt{\{document\_text\}} placeholder, which was replaced by the input document at inference time.

The full prompt templates for all workloads are shown below.

\vspace{1.2em}
\noindent\textsc{AG News}
\vspace{0.5em}
\begin{lstlisting}[caption={\textsc{agnews} classification prompt},label={lst:agnews}]
I will give you a news article. Here is the article: {document_text}

Assign the article to one of the following categories and return only the corresponding number:
- 0 = World: The article primarily discusses international news, global events, diplomacy, conflicts, or issues involving multiple countries or regions.
- 1 = Sports: The article is mainly about sporting events, teams, athletes, competitions, results, or sports-related news.
- 2 = Business: The article focuses on economic matters, companies, markets, finance, industry trends, or business-related developments.
- 3 = Sci/Tech: The article covers topics in science, technology, research, discoveries, innovations, or advancements in scientific or technological fields.

You must respond with ONLY the number (0, 1, 2, or 3) that best matches the main topic of the article.
\end{lstlisting}

\vspace{1.2em}
\noindent\textbf{Enron}
\vspace{0.5em}
\begin{lstlisting}[caption={\textsc{enron} sender classification prompt},label={lst:enron}]
I will give you an email. Here is the email: {document_text}

Your task is to determine if this email was sent from a senior executive or other high-ranking person at Enron.

- Return True if the email was sent from a senior executive (CEO, President, VP, Director, etc.) or other high-ranking person.
- Return False if the email was sent from a lower-level employee or non-executive.

Important notes: Look for job titles, positions, and signatures that indicate seniority. Consider both formal titles and contextual clues about the sender's role and authority level.

You must respond with ONLY True or False.
\end{lstlisting}

\vspace{1.2em}
\noindent\textbf{Court Opinion}
\vspace{0.5em}
\begin{lstlisting}[caption={\textsc{court} opinion reversal prompt},label={lst:court}]
I will give you a Supreme Court opinion. Here is the opinion: {document_text}

Your task is to determine whether this court opinion reverses a lower court's ruling.

Carefully read the opinion and consider the following:

- Return True if the Supreme Court (or the relevant higher court) reverses the decision of a lower court.
- Return False if the Supreme Court upholds (affirms) the lower court's ruling, or if the opinion does not address a lower court's decision.

You must respond with ONLY True or False.
\end{lstlisting}

\vspace{1.2em}
\noindent\textbf{PubMed}
\vspace{0.5em}
\begin{lstlisting}[caption={\textsc{pubmed} study type prompt},label={lst:pubmed}]
I will give you a full biomedical research article from PubMed. Here is the article: {document_text}

Your task is to determine the type of biomedical study described in the full article.

Carefully read the article and determine which of the following study types best describes the research. Consider the study's methodology, data sources, and overall approach. Choose the single most appropriate type from the list below and return only the corresponding number:
- 0 = Randomized Controlled Trial (RCT): Participants are randomly assigned to groups to compare outcomes.
- 1 = Observational Study: Researchers observe existing groups without assigning interventions (includes cohort, case-control, cross-sectional).
- 2 = Meta-analysis or Systematic Review: Combines and analyzes results from multiple prior studies using systematic methods.
- 3 = Bench / Wet-lab Experimental Study: Laboratory-based experiments (e.g., cell culture, animal models, in vitro assays).
- 4 = Computational / Bioinformatics Study: Uses computational models, simulations, or large-scale data analysis (e.g., genomics, proteomics).
- 5 = Narrative Review (non-systematic): Describes a topic broadly without a structured or systematic review process.

You must respond with ONLY the number (0-5) that best matches the article's main study type.
\end{lstlisting}

\vspace{1.2em}
\noindent\textbf{Game Reviews}
\vspace{0.5em}
\begin{lstlisting}[caption={\textsc{game} review comparison prompt},label={lst:games}]
I will give you a review for a game. Here is the review: {document_text}

Your task is to carefully read the following review and decide whether it mentions any other games in a more positive way than the game being reviewed.

Consider whether the reviewer compares the current game to another game and expresses a preference for the other game, either directly or indirectly. Look for statements that praise another game or suggest that the other game is better in some respect.

- Return True if the review references another game and describes it more favorably than the game being reviewed.
- Return False if the review does not mention other games, or if it does not express a preference for another game over the current one.

You must respond with ONLY True or False.
\end{lstlisting}

\vspace{1.2em}
\noindent\textbf{Legal Docs}
\vspace{0.5em}
\begin{lstlisting}[caption={\textsc{legal} IP clause prompt},label={lst:legal}]
I will give you a legal document. Here is the document: {document_text}

Your task is to determine if this document contains any type of covenant not to sue or agreement not to challenge intellectual property rights. This includes both direct promises and indirect restrictions.

- True if it contains ANY of these:
  - Agreement not to contest/challenge IP validity or ownership
  - Promise not to question/attack/impugn IP rights
  - Agreement not to take actions inconsistent with IP ownership
  - Covenant not to bring claims/suits related to IP
  - More generally, any provision that could be interpreted as a restriction on future IP challenges
- False if it contains none of the above

You must respond with ONLY True or False.
\end{lstlisting}

\vspace{1.2em}
\noindent\textbf{Wiki Talk}
\vspace{0.5em}
\begin{lstlisting}[caption={\textsc{wiki\_talk} edit revert prompt},label={lst:wikitalk}]
I will give you a Wikipedia Talk page discussion. Here is the discussion: {document_text}

Your task is to carefully read the following discussion and determine the outcome regarding the edits in question.

Consider whether the discussion led to a reversion (a rollback of previous edits) or resulted in a stable change to the content.

- Return True if the discussion resulted in reverting or rolling back changes to a previous version.
- Return False if the discussion led to stable changes being kept, or if no changes were made as a result of the discussion.

Be sure to look for explicit mentions of reversion, rollback, or restoration of prior content, as well as consensus to keep new changes.

You must respond with ONLY True or False.
\end{lstlisting}

\vspace{1.2em}
\noindent\textbf{FEVER}
\vspace{0.5em}
\begin{lstlisting}[caption={\textsc{fever} claim verification prompt},label={lst:fever}]
I will give you a claim and a list of documents that may or may not explicitly support the claim. Here is the claim and documents: {document_text}

Your task is to assess whether the provided claim is supported by the accompanying documents.

- Return True if at least one of the documents clearly supports the claim.
- Return False if none of the documents support the claim, or if the evidence is unclear or insufficient to determine support.

You must respond with ONLY True or False.
\end{lstlisting}

%% file: revisions/paramsensitivity.tex
\rone{Task cascades depend on three main hyperparameters: (1) the number of surrogate operations proposed per iteration ($n_s$), (2) the number of agentic refinement iterations ($n_a$), and (3) the set of document fractions to consider ($\mathcal{F}$). To evaluate robustness to these choices, we conduct a sensitivity analysis on two representative workloads: \textsc{enron} (where task cascades achieve large gains) and \textsc{games} (where gains are modest). For each workload, we vary one parameter at a time while holding others at their default values: $n_s \in \{3, 5, 10\}$ with $n_a = 3$ fixed, $n_a \in \{1, 2\}$ with $n_s = 5$ fixed, and $\mathcal{F} \in \{\{0.25, 1.0\}, \{0.25, 0.5, 1.0\}\}$ with $n_s = 5, n_a = 1$. All experiments use a 90\% accuracy target and report costs relative to the \method{2-Model Cascade} baseline.}

\rone{\Cref{fig:param-sensitivity} shows the results. \method{Task Cascades}  consistently outperform baselines across all parameter configurations on both workloads. On \textsc{enron}, all settings achieve substantial cost savings over the \method{2-Model Cascade} baseline while meeting the accuracy target. On \textsc{games}, where \method{Task Cascades} shows only modest improvements over \method{2-Model Cascade}, parameter variations remain stable, demonstrating robustness even in challenging scenarios.}

\rone{The sensitivity analysis sheds light on several practical insights. First, $n_s$ (the number of surrogate operations to explore) has a more pronounced effect on cost than $n_a$ (the number of refinement iterations). This suggests it is worth generating multiple diverse surrogate candidates per iteration rather than relying heavily on iterative refinement. Second, all tested document fraction sets yield similar performance, indicating that the approach is robust to this choice. However, it is important to include the full document fraction (1.0) in $\mathcal{F}$, as some tasks or proxy models may benefit from complete context. We recommend spanning a range of fractions---from a small portion (e.g., 0.1 or 0.25) to the full document---to enable effective document pruning while maintaining flexibility. Overall, these findings suggest that modest parameter values ($n_s \geq 3$, $n_a \geq 1$, and $\mathcal{F}$ spanning from small to full document) are sufficient for strong performance, though users can increase $n_s$ when additional optimization budget is available.}

%% file: revisions/opcosttable.tex
\begin{table*}
\centering
\vspace{-10pt}
  \begin{minipage}{0.99\textwidth}
  \centering
  \caption{\all{Offline optimization costs, break-even analysis, and estimated cost at 1M documents.
Left columns (Optimization Cost) report the US dollars required to build each method (\method{TC}, \method{TC (Lite)}, \method{2MC}), given $D_\mathrm{dev}=200$ (except 150 for \textsc{legal}). 
Middle columns (Break-Even) show how many inference documents must be processed before total cost becomes lower than running the \method{Oracle Only} on every document.
Rightmost columns estimate total end-to-end cost for processing 1 million documents: offline optimization + per-document inference. Parentheses show, for each \method{TC} variant, its ratio to the \method{2MC} baseline.}}
\vspace{-10pt}
  \label{tab:optimization-and-breakeven}
  \begin{tabular}{@{}lrrrrrrrrr@{}}
  \toprule
  & \multicolumn{3}{c}{\textbf{Optimization Cost (\$)}} 
  & \multicolumn{3}{c}{\textbf{Break-Even (Documents)}} 
  & \multicolumn{3}{c}{\textbf{Estimated Total Cost @ 1M Docs (\$)}} \\
  \cmidrule(lr){2-4} \cmidrule(lr){5-7} \cmidrule(lr){8-10}
  \textbf{Workload}
  & \textbf{\method{TC}}
  & \textbf{\method{TC (Lite)}}
  & \textbf{\method{2MC}}
  & \textbf{\method{TC}}
  & \textbf{\method{TC (Lite)}}
  & \textbf{\method{2MC}}
  & \textbf{\method{TC}}
  & \textbf{\method{TC (Lite)}}
  & \textbf{\method{2MC}} \\
  \midrule
  \textsc{agnews}     & \$3.32  & \$0.58 & \$0.12 & 6,173.87 (27.0$\times$) & 1,080.54 (4.7$\times$) & 228.57 & \$28.07 (0.75$\times$) & \$29.83 (0.79$\times$) & \$37.62 \\
  \textsc{court}      & \$29.51 & \$7.59 & \$2.70 & 2,615.32 (10.0$\times$)  & 680.85 (2.6$\times$)   & 262.14 & \$2{,}558.51 (0.73$\times$) & \$2{,}677.09 (0.76$\times$) & \$3{,}515.20 \\
  \textsc{enron}      & \$13.78 & \$3.38 & \$1.15 & 1,757.40 (10.9$\times$) & 431.84 (2.7$\times$)   & 160.56 & \$97.65 (0.13$\times$) & \$110.13 (0.14$\times$) & \$763.65 \\
  \textsc{fever}      & \$39.52 & \$10.27& \$3.69 & 2,513.60 (10.4$\times$) & 656.72 (2.7$\times$)   & 241.77 & \$579.52 (0.58$\times$) & \$640.27 (0.64$\times$) & \$1{,}003.69 \\
  \textsc{games}      & \$10.92 & \$2.61 & \$0.87 & 3,831.92 (12.9$\times$) & 808.24 (2.7$\times$)   & 296.17 & \$1{,}073.67 (1.10$\times$) & \$685.11 (0.70$\times$) & \$975.87 \\
  \textsc{legal}      & \$45.22 & \$11.89& \$4.30 & 1,745.22 (7.9$\times$)  & 482.16 (2.2$\times$)   & 222.09 & \$1{,}748.54 (0.28$\times$) & \$2{,}598.41 (0.41$\times$) & \$6{,}312.89 \\
  \textsc{pubmed}     & \$25.22 & \$6.44 & \$2.28 & 2,557.29 (10.9$\times$) & 653.78 (2.8$\times$)   & 233.85 & \$613.22 (0.87$\times$) & \$608.44 (0.87$\times$) & \$702.28 \\
  \textsc{wiki\_talk} & \$9.49  & \$2.23 & \$0.72 & 2,687.62 (12.9$\times$) & 632.54 (3.0$\times$)   & 208.70 & \$153.49 (0.68$\times$) & \$155.23 (0.69$\times$) & \$225.72 \\
  \bottomrule
  \end{tabular}
  \end{minipage}%
\end{table*}

\begin{table}
\centering
\begin{minipage}{0.99\linewidth}
\centering
\caption{\all{Latency measurements (mean~$\pm$~stdev, in seconds) for optimization and inference (inf.). 
Due to availability of credits, \method{2MC}, \method{TC (Lite)}, and \method{Oracle Only} used Azure OpenAI APIs, 
while \method{TC} used standard OpenAI endpoints (Azure is $\sim$2.5$\times$ slower~\cite{azure_latency_reddit}). 
Standard deviations are high for all methods, due to API tail latencies.
Rightmost column estimates total inference latency for processing 1M documents of API requests, in hours.
Estimations represent a lower bound, since they do not account for API tail latencies.}}
\label{tab:latency}
\footnotesize
\begin{tabular}{@{}lccc@{}}
\toprule
\textbf{Method} 
& \textbf{Optimization (s)} 
& \textbf{Inf. @ $\sim$1K docs (s)} 
& \textbf{Est.\ Inf. @ 1M docs (h)} \\
\midrule
\methodsmall{2MC} & 64.93 $\pm$ 57.60   & 280.64 $\pm$ 291.36   & $3.0$ $\pm$ $0.1$ \\
\methodsmall{TC}  & 911.45 $\pm$ 381.23 & 201.19 $\pm$ 149.03   & $2.2$ $\pm$ $0.0$ \\
\methodsmall{TC (Lite)} & 1253.0 $\pm$ 947.3 & 175.4 $\pm$ 90.6     & $1.9$ $\pm$ $0.0$ \\
\methodsmall{Oracle Only} & -- & 26.14 $\pm$ 56.80 & $0.3$ $\pm$ $0.0$ \\
\bottomrule
\end{tabular}
\end{minipage}%
\vspace{-5pt}
\end{table}

%% file: revisions/opcost.tex
\all{\section{Optimization Cost and Latency Analysis}
\label{sec:optimization-cost-analysis}}

\all{In this section, we provide detailed measurements of the offline optimization costs and latencies for both \method{Task Cascades} and \method{2-Model Cascade}, as well as analysis showing when the upfront optimization investment is recovered through inference savings and total costs at scale.}

\all{\subsection{Optimization Cost and Break-Even Analysis}}

\all{\Cref{tab:optimization-and-breakeven} reports optimization costs, break-even points, and estimated end-to-end costs for processing 1 million documents for each workload. We define the break-even point as the number of inference documents required for the total cost of a method---its one-time optimization cost plus per-document inference---to become lower than simply running the oracle model on every document.}

\topic{Optimization Cost}
\all{As explained in \Cref{sec:cost-model}, \method{Task Cascades} incurs higher optimization costs because it performs two additional procedures: {\em (i)} training a lightweight relevance classifier for document restructuring, and {\em (ii)} evaluating multiple candidate rewrites of the pipeline on the development set during cascade assembly. In contrast, \method{2-Model Cascade} only evaluates the proxy and oracle models to set confidence thresholds. The magnitude of this cost depends on document length and task complexity: longer documents (\textsc{fever}, \textsc{legal}) increase oracle labeling costs, while harder tasks require more surrogate generation rounds.}

\all{We also include \method{Task Cascades (Lite)}, a lightweight variant that does not consider candidate tasks that leverage the oracle model, as described in \Cref{sec:setup}. This configuration reduces optimization cost by roughly 70–80\% (2.2$\times$--4.7$\times$ the \method{2-Model Cascade} baseline, mean 2.9$\times$) while maintaining similar inference behavior.}

\topic{Break-Even Analysis} \all{Across workloads, \method{Task Cascades} reaches break-even between 1{,}745 documents (\textsc{legal}) and 6{,}174 documents (\textsc{agnews}), with a mean of 2{,}986. \method{Task Cascades (Lite)} achieves proportionally lower break-even thresholds (480–1{,}080 documents, mean 678) due to its smaller optimization cost. The ratio between \method{Task Cascades} and \method{2-Model Cascade} break-even points ranges from 7.9$\times$ (\textsc{legal}) to 27.0$\times$ (\textsc{agnews}), with an overall mean of 11.4$\times$---reflecting the varying magnitude of inference savings.}

\topic{Costs at Scale}
\all{The rightmost columns of \Cref{tab:optimization-and-breakeven} estimate total end-to-end cost for processing 1 million documents (optimization + inference). At this scale, optimization costs become negligible. \method{Task Cascades (Lite)} provides cost savings over \method{2-Model Cascade} on all 8 workloads, with median savings of \$327 per million documents and up to \$3{,}714 on \textsc{legal}. \method{Task Cascades} provides savings on 7 of 8 workloads (all except \textsc{games}), with median savings of \$257 and up to \$4{,}564 on \textsc{legal}. On \textsc{games}, \method{2-Model Cascade} fails to meet the 90\% accuracy target (\Cref{tab:results-90-revised}), while both \method{Task Cascades} variants do.}

\topic{Practical Implications}
\all{Most workloads in our evaluation contain tens of thousands to millions of documents (\Cref{tab:dataset-stats}), making these optimization costs negligible at scale. For example, running on the full \textsc{enron} corpus ($\sim$500k documents) requires only 0.35\% of documents to reach break-even, while the smaller \textsc{court} dataset ($\sim$36k documents) reaches break-even after processing 7.3\% of its corpus. These results show that \method{Task Cascades}---and particularly its \method{Lite} variant---recovers its investment quickly in realistic production settings.}

\all{\subsection{Latency Measurements}}

\all{\Cref{tab:latency} reports optimization and inference latencies for all methods, as well as estimated total inference latency for processing 1 million documents under 32-way concurrency. Measurements were taken on a 2024 MacBook Air (M4 chip) using 32 parallel API threads, with exponential back-off sleeps to handle rate limits. Note that API latencies can vary significantly across providers and time of day. \method{Task Cascades} optimization takes $\sim 14 \times$ longer than \method{2-Model Cascade}, largely due to surrogate generation (o1-mini) and evaluating many candidate configurations. However, runtimes are not directly comparable across providers: as noted in \Cref{tab:latency}, \method{2-Model Cascade}, \method{Oracle Only}, and \method{Task Cascades (Lite)} were measured on Azure OpenAI while \method{Task Cascades} used standard OpenAI endpoints; Azure can be $\sim$2.5$\times$ slower~\cite{azure_latency_reddit}. This provider effect could explain why \method{Task Cascades (Lite)} shows a higher mean optimization runtime (1253s) despite its lower token cost.}

\all{For inference at scale at the 1M mark, our per-document measurements suggest that \method{Task Cascades} and \method{Task Cascades (Lite)} achieve within the same order of magnitude of latency (hours) to \method{2-Model Cascade}. However, these estimates should be interpreted cautiously: they do not account for API tail latencies or rate limiting, and the measurements were taken across different providers (Azure vs. standard OpenAI) with known latency differences. In batch processing scenarios---the primary use case for large-scale document processing---higher optimization latency is acceptable when it yields substantial cost savings, as demonstrated by the widespread adoption of batch APIs with extended completion times between 24 and 72 hours~\cite{openai_batch_api, google_vertex_batch_gemini}.}